\documentclass[sigplan, screen]{acmart}
\setcopyright{none}

\usepackage[normalem]{ulem}
\usepackage{verbatim}
\usepackage{xcolor}
\usepackage[T1]{fontenc}
\usepackage{epsfig,endnotes,footnote}
\usepackage{array,multirow}
\usepackage{amsmath}
\usepackage{enumitem}
\usepackage{booktabs}
\usepackage{xspace}
\usepackage{relsize}
\usepackage{hyphenat}
\usepackage{algorithm}
\usepackage{algorithmicx}
\usepackage{algpseudocode}
\usepackage{rotating}
\usepackage{graphicx}
\usepackage{subfig}
\DeclareGraphicsExtensions{.pdf,.jpg,.png}
\usepackage{multirow}

\usepackage{url}
\usepackage{microtype}

\usepackage{tabularx}
\usepackage{mathpartir}

\hypersetup{
    bookmarks=true,         
    unicode=false,          
    pdftoolbar=true,        
    pdfmenubar=true,        
    pdffitwindow=true,      
    pdftitle={},    
    pdfauthor={},     
    pdfsubject={},   
    pdfnewwindow=true,      
    pdfkeywords={}, 
    colorlinks=true,       
    linkcolor=blue,          
    citecolor=blue,        
    filecolor=magenta,      
    urlcolor=cyan           
}
\newcommand{\system}{Ocelot\xspace}
\newcommand{\sys}{Ocelot\xspace}
\newcommand{\m}[1]{\mathsf{#1}}
\newcommand{\mt}[1]{\mathit{#1}}
\newcommand{\ft}[1]{\textit{#1}}
\newcommand{\paragraphb}[1]{\vspace{0.05in}\noindent{\bf #1}\xspace}
\newcommand{\bnfdef}{::=}
\newcommand{\bnfalt}{\,|\,}
\newcommand{\rulename}[1]{\textsc{#1}}
\newcommand{\ifthen}[3]{\m{if}\ #1\ \m{then}\ #2\ \m{else}\ #3}

\newcommand{\elet}[3]{\m{let}\ #1 = #2 ~\m{in}\ #3}

\newcommand{\fresh}{$\m{Fresh}$\xspace}
\newcommand{\consistent}{$\m{Consistent}$\xspace}

\newcommand{\bop}{\odot}
\newcommand{\uop}{\oslash}

\newcommand{\astart}[1]{\m{start_{atom}(#1)}}
\newcommand{\aend}{\m{end_{atom}}}
\newcommand{\atomic}[2]{\astart{#1};#2;\aend}

\usepackage{eqparbox}
\newdimen{\algindent}
\algnewcommand\LeftComment[2]{%
\hspace{#1\algindent}$\triangleright$ \eqparbox{COMMENT}{#2} \hfill %
}

\newcommand{\alg}[1]{\ensuremath{\mathit{#1}}}

\newcommand{\timestamp}{\tau}
\newcommand{\context}{\kappa}

\newcommand{\actx}{\kappa_\mt{atom}}
\newcommand{\jctx}{\kappa_\mt{jit}}
\newcommand{\nvmem}{N}
\newcommand{\vmem}{V}

\newcommand{\cmd}{c}

\newcommand{\depth}{\mt{n_{atom}}}

\newcommand{\stepsto}{\longrightarrow}
\newcommand{\Stepsto}[1]{\stackrel{#1}{\Longrightarrow}}
\newcommand{\SeqStepsto}[1]{\stackrel{#1}{\longrightarrow}}

\newcommand{\MSeqStepsto}[1]{\stackrel{#1}{\longrightarrow^*}}

\newcommand{\ulog}{\mathcal{L}}

\newcommand{\fdecls}{\textit{FD}\xspace}
\newcommand{\pdecls}{\textit{PD}\xspace}
\newcommand{\fsums}{\textit{FS}\xspace}
\newcommand{\pmap}{\textit{PM}\xspace}
\newcommand{\pol}{\textit{pol}\xspace}
\newcommand{\instr}{\iota}
\newcommand{\aid}{\textit{aID}\xspace}
\newcommand{\pid}{\textit{pID}\xspace}

\newcommand{\lgvec}[1]{\overrightarrow{#1}}

\newcommand{\hole}{[\,]}
\newcommand{\prov}{\rho}
\newcommand{\provs}{\mathcal{P}}
\newcommand{\ins}{\mathcal{I}}
\newcommand{\cc}{\textit{cc}}
\newcommand{\notes}[2]{} 
\newcommand{\limin}[1]{\notes{magenta}{Limin says: #1}}
\newcommand{\todo}[1]{\notes{red}{TODO: #1}}
\newcommand{\milijana}[1]{\notes{blue}{Milijana says: #1}}

\newcommand{\extra}[1]{}

\newif\ifproofs
\proofstrue

\newtheorem{thm}{Theorem}
\newtheorem{lem}[thm]{Lemma}

\newtheorem{defn}[thm]{Definition}

\graphicspath{{graphics/}}

\begin{document}

\title{
    Automatically Enforcing Fresh and Consistent Inputs in Intermittent Systems}

\date{}

\author{Milijana Surbatovich}

\affiliation{
  \institution{Carnegie Mellon University}            
  \city{Pittsburgh}
  \state{PA}
  \country{USA}                    
}
\email{milijans@andrew.cmu.edu}          

\author{Limin Jia}

\affiliation{
  \institution{Carnegie Mellon University}            
  \city{Pittsburgh}
  \state{PA}
  \country{USA}                    
}
\email{liminjia@andrew.cmu.edu}          

\author{Brandon Lucia}

\affiliation{
  \institution{Carnegie Mellon University}            
  \city{Pittsburgh}
  \state{PA}
  \country{USA}                    
}
\email{blucia@andrew.cmu.edu}          

\begin{abstract}
    Intermittently powered energy-harvesting devices enable new
 applications in inaccessible environments. Program executions must be
 robust to unpredictable power failures, introducing new challenges in
 programmability and correctness. One hard problem is that input
 operations have implicit constraints, embedded in the behavior of
 continuously powered executions, on \emph{when} input values can be
 collected and used.
This paper aims to develop a formal framework for enforcing these constraints.
We identify two key
properties---\emph{freshness} (i.e., uses of inputs must satisfy the
same time constraints as in continuous executions) and
\emph{temporal consistency} (i.e., the collection of a set of inputs must satisfy 
the same time constraints as in continuous executions).
We formalize these properties and show that they can be enforced using \emph{atomic}
regions.  We develop Ocelot, an LLVM-based analysis and transformation
tool targeting Rust, to enforce these properties automatically.  Ocelot provides the
programmer with annotations to express these constraints
and infers atomic region placement in a program to 
satisfy them. We then formalize Ocelot's
design and show that Ocelot generates correct programs with little
performance cost or code changes.

    \end{abstract}


\maketitle
\thispagestyle{empty}

\section{Introduction}
\todo{Tweaks to emphasize expressing correctness}
Energy-harvesting computer systems collect their operating energy from the
environment, enabling autonomous operation over long periods of time without the
need for battery maintenance.  The key challenge of
energy-harvesting systems is that power fails if there is insufficient energy
to harvest. 
When an energy-harvesting system runs software,
a power interruption may impede forward progress~\cite{mementos,quickrecall},
leave memory state inconsistent~\cite{dino,ratchet}, leave I/O state~\cite{sytare,restop} 
or data~\cite{ibis,formal-foundations} inconsistent with execution state, or leave
I/O data inconsistent with a device's environment~\cite{mayfly,tics}. 

Intermittent execution~\cite{dino} of software enables sophisticated
computation on energy-harvesting systems, leveraging
tightly integrated non-volatile memory to retain
state across failures. There are 
many 
approaches to  address the software reliability challenges of
intermittent computing.  Most prior efforts 
focus primarily on problems related to progress and memory consistency.  To save
state, these techniques rely on in-code checkpointing (or
tasks)~\cite{mementos,dino,alpaca,chain,ratchet,chinchilla,ink,coati}, or rely on a
dynamic ``just-in-time'' (JIT) checkpointing
mechanism~\cite{hibernus,hibernusplusplus,samoyed,quickrecall,idetic,
catnap,forgetfailure,reliable-time} that
captures a snapshot of volatile state just before power fails. 

Most intermittent computing happens on sensor-enabled devices destined
for deeply-embedded deployment, where I/O drives the computation.
Fortunately, recent work has begun investigating the unique challenges of
I/O in intermittent systems. 
Some work ensures the basic, correct operation of peripherals and their drivers
across power failures~\cite{sytare,restop,samoyed,karma, adp, auto-io}, avoiding crashes, hangs,
and driver state corruption.  Other work addresses subtle
interactions between I/O and checkpointing that lead to data corruption~\cite{ibis,formal-foundations}. 
These efforts 
enable correct basic operation of I/O devices in an intermittent execution.
Operating in the real world, however, places correctness requirements on an
intermittent system that go beyond ensuring that drivers and data 
remain uncorrupted. 

Unlike a continuously-powered execution, an intermittent execution may violate
implicit constraints on {\em when} inputs should be collected and used, due
to the unpredictable time spent recharging after a power failure.
An intermittent execution may use an input that is too old (i.e., stale) 
if the system checkpoints after the input is collected, but power 
fails before it is used.  The need to avoid use of stale inputs is a {\em
  freshness} requirement. 
Some programs require {\em multiple} input values to be sampled together (e.g.,
a pressure and a humidity reading) so that they come from a
consistent point in time.  The need to ensure that multiple inputs are sampled
together is {\em temporal consistency}, which is violated by a 
checkpoint and power failure between these readings.

Freshness and temporal consistency belong to the broader category of {\em
timeliness} requirements on inputs. Prior work explored timely
intermittent execution~\cite{mayfly,tics,reliable-time,tardis} but lacks formally specified 
correctness conditions.  Existing approaches rely on the addition of 
hardware to track time during power failures and often
require writing extra code to mitigate the misuse of expired inputs.  
 Moreover, existing work focuses on freshness (e.g., using inputs before they expire) 
 and does little
to enforce temporal consistency.   

In this work, we introduce formal definitions of freshness and 
temporal consistency and develop \sys,
which automatically enforces 
specified timing constraints
in intermittent systems without needing timekeeping hardware. 
\sys gives the programmer constructs to specify 
what timing properties matter for their program
and enforces that specification by leveraging atomicity, 
generating programs that are correct-by-construction. 
Instead of enforcing programmer-specified expiration times, \sys enforces freshness and temporal
consistency by ensuring that an intermittent execution does what some continuous execution
would do; {\em the continuous execution is the specification of correct behaviour}.
\sys asks the programmer to express freshness and temporal consistency
requirements only and asks neither for timing specification on the collection
or use of inputs, nor for mitigation actions to handle expired inputs.  
\sys's atomic region inference algorithm then automatically inserts 
atomic regions that contain input-derived variable definitions and uses.
If power fails during an atomic region, the region re-executes (idempotently)
from the start.
Outside of an atomic region, the system defaults to a baseline intermittent
execution model (i.e., in our work, JIT checkpoints~\cite{samoyed,hibernus}).

We formalize this notion of freshness and temporal
consistency using a modeling language and investigate how to prove our design correct. 
We implement \sys for Rust using analyses built in LLVM~\cite{llvm}. We evaluate our implementation on a real energy-harvesting hardware platform~\cite{capybara} using a
collection of applications taken from prior work, and a new tire safety monitoring
application that we developed. 
Our results show that \sys effectively identifies atomic regions that enforce
both freshness and temporal consistency.  \sys imposes less than 10\% runtime overhead compared to both JIT
checkpoints and to atomic regions implementations from prior work. 
\sys demands less of the programmer, compared to two systems from prior work that
address I/O timeliness~\cite{tics,samoyed}. 
Most importantly, \sys provides a formally defined correctness criterion
for collection and use of intermittent inputs, which no prior system provides.

To summarize, the main contributions are:
\begin{itemize}
\item We provide the first formal definitions for 
{\em freshness} and {\em temporal consistency} and show that 
atomicity is sufficient to enforce these properties.

\item We develop \sys, an analysis that inserts atomic regions 
to enforce these properties without asking the programmer to think about 
real time, mitigations for timeliness failures, or added hardware.
\item We prototype \sys for Rust and use it to add atomic regions to a set 
of programs from prior work and a tire monitoring program we developed.

\item We evaluate \sys on real energy-harvesting hardware and show that its atomic regions ensure
  these properties at little runtime or programming overhead.
\end{itemize}

\section{Background and Motivation}
\label{sec:background}

Software executes intermittently on energy-harvesting systems, relying
on system support to ensure progress and memory correctness despite power
failures.  I/O complicates an intermittent system, requiring additional
correctness reasoning to ensure both correct device operation \emph{and}
the freshness and temporal consistency properties addressed by \sys.

\subsection{The Basics of Intermittent Computing}
Software executing intermittently on an energy-harvesting system makes forward
progress only as sufficient energy is available. We show this 
in the graph in Figure~\ref{fig:exec}, top.  A system collects energy
using, e.g., a solar panel or radio wave collector, storing small amounts
of energy in a tiny battery or capacitor (red segments).  After a system-specific amount of
energy accumulates, hardware activates the system to begin executing, quickly
consuming the energy (green segments).  The executing system may collect sensor inputs, 
run computations (e.g., machine learning to process sensor
data~\cite{sonic,manic,camaroptera}) on an ultra-low-power CPU or
microcontroller, and log or transmit results via a wireless radio
link.
\begin{figure}[bp]
  \centering
  \includegraphics[width=0.9\columnwidth]{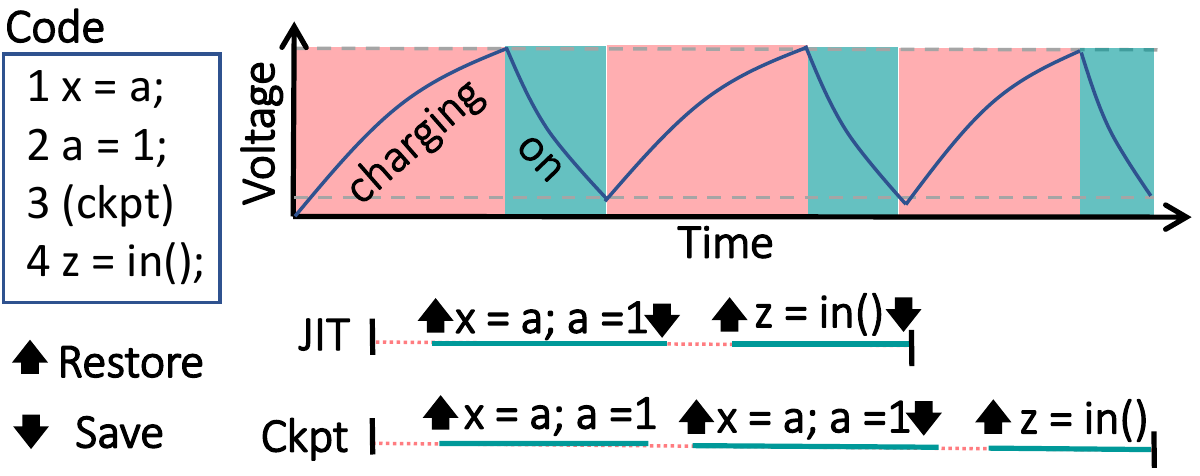}
\caption{JIT and Checkpoint based intermittent execution}
  \label{fig:exec}
\end{figure}

Prior work identified and addressed several
progress~\cite{mementos,quickrecall,idetic} and
correctness~\cite{dino,ratchet,clank,quickrecall,mementos,idetic,alpaca,chain,coati,
samoyed,hibernus,hibernusplusplus,ink,catnap}
challenges to intermittent execution.  The main idea in these works is to
ensure that non-volatile memory remains consistent as execution proceeds in
bursts.
There are two broad classes of solutions,
``just-in-time'' checkpointing systems~\cite{samoyed, hibernus,
quickrecall} and checkpoint- (or ``task-'') based
systems~\cite{dino,ratchet,alpaca,coati}. We illustrate the difference 
using the code snippet in Figure~\ref{fig:exec}.  A JIT checkpointer uses hardware to
monitor energy. The software runtime backs up volatile state (registers, stack) just before power fails.  On reboot,
the  
system restores volatile state and continues.
In the example, power fails after executing line 2. On reboot, the execution resumes from line 4.   
A checkpoint-based system encounters explicit code points where 
 it collects a checkpoint and continues executing, such as line 3.  After a power
failure, the system resumes from the last saved explicit checkpoint. If power fails 
after executing line 2, the execution restarts from line 1. 
The checkpoint saves volatile state, like JIT, but also saves some
non-volatile state to ensure that they remain consistent.  Prior work showed
that a checkpoint must back up non-volatile memory that will be
accessed after the checkpoint first by a read, then by a write, i.e., a
Write-After-Read (WAR) dependence~\cite{dino,ratchet}, i.e., {\tt a} in the 
example. 

Inputs complicate checkpointing.  
Unlike checkpoint-based systems, a JIT system
never re-executes code after reboot.  
In some cases, however, correctness {\em requires} re-executing to re-collect
an input; in such cases JIT checkpointing always renders execution incorrect.  
Prior work showed that checkpointing causes incorrect behavior if a 
value derived from an input is not correctly backed up 
~\cite{ibis,formal-foundations}.  
To avoid the incorrect behavior, a system must add to the checkpoint
conditionally-written, non-volatile data not already checkpointed due to a WAR
dependence (the ``exclusive may-write'' or EMW
set~\cite{ibis,formal-foundations}).
Even after resolving these memory consistency issues, inputs still complicate
intermittent correctness, because of input {\em timeliness
constraints}~\cite{mayfly,tics,capybara,tardis}.

\subsection{Inputs Violating Freshness and Consistency}

\begin{figure}[tbp]
  \centering
  \includegraphics[width=1.0\columnwidth]{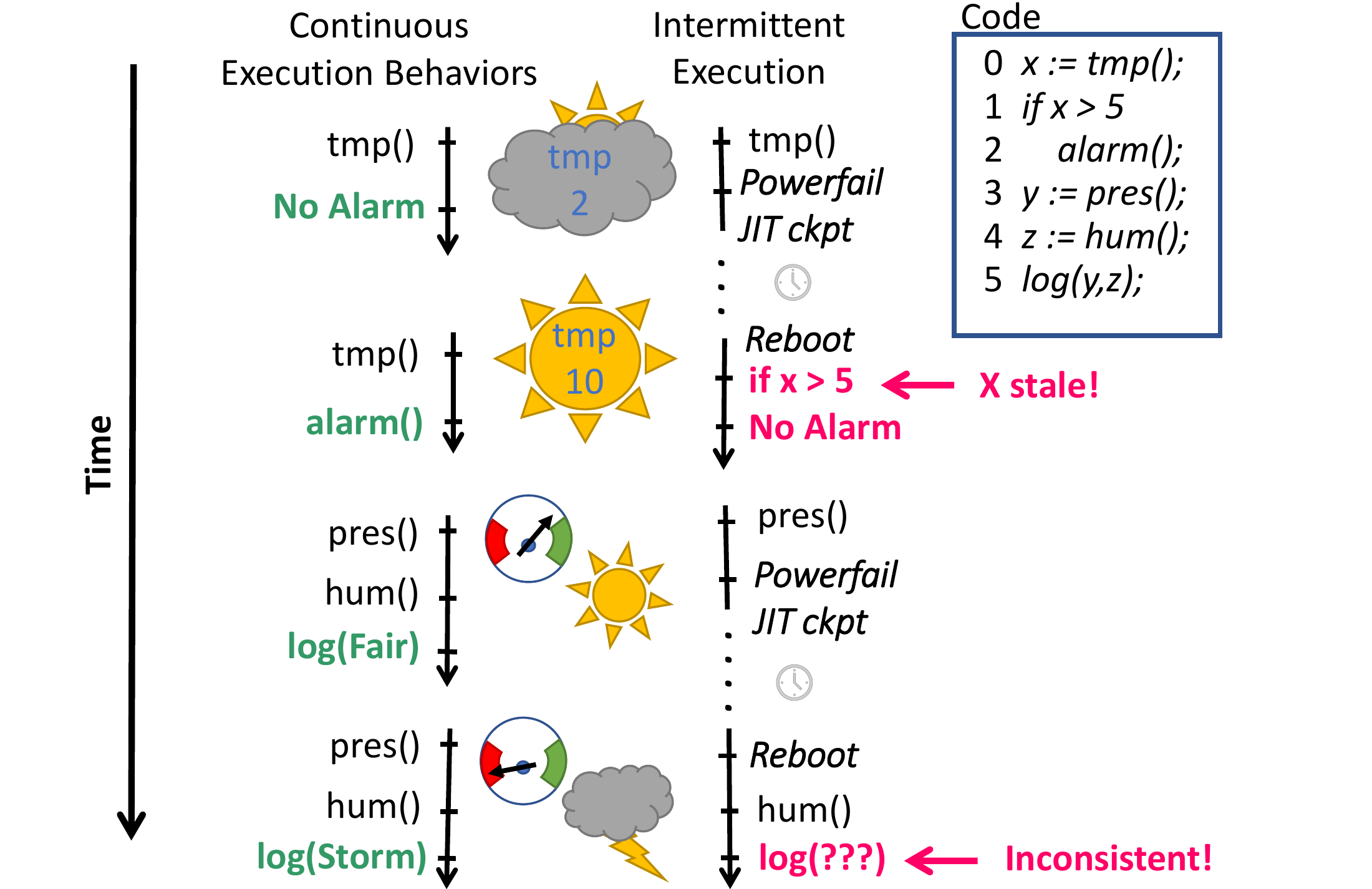}
\caption{Freshness and temporal consistency problems.
}
  \label{fig:problem}
\end{figure}

Intermittent execution can violate {\em freshness} and {\em temporal
consistency}, which are implicit correctness constraints illustrated in
Figure~\ref{fig:problem}.
The example program reads a thermometer, raising an alarm for high 
temperatures. 
The program then logs pressure and humidity sensor data that may indicate
a storm.   
Time flows down and at left are possible continuous executions, each
corresponding to the weather in the middle.  At right are intermittent
executions (assuming JIT checkpointing).
Power fails between instructions 0 and 1 and between 3 and 4, spending arbitrary
time while powered off.

\paragraphb{Violating Freshness} 
The time delay of a power failure violates data freshness, causing 
incorrect behavior if the temperature changes during the delay.
The continuous execution raises the alarm at high temperature.
The intermittent execution, however, senses cold, then checkpoints and powers off. 
On reboot, it raises no alarm, even though the temperature is high.
The code implicitly requires the use of $x$ while it is {\em fresh}, but the power failure prevents this.  
For an intermittent execution to match a continuous execution, power must not fail
between sensing the temperature and executing the branch on $x$.

\paragraphb{Violating Temporal Consistency}  
The time delay of a power failure may compromise the {\em temporal consistency} of a collection of sensor data.
With initially fair weather that becomes stormy, a continuous execution may sense 
high pressure and 
low humidity (i.e., no storm), or sense low pressure and high humidity 
(i.e., a storm), logging either condition. 
The intermittent execution, however, reads high pressure before power fails
(fair weather), and high humidity after rebooting (storm). 
The sensed values are inconsistent with the fair or stormy weather seen by continuous executions. 
For intermittent execution to match continuous execution, power must not 
fail between the pressure and humidity readings.

\subsection{Prior Approach: Timeliness}
Freshness and temporal consistency are correctness conditions on 
{\em when} data from input operations may be used, similar but distinct from 
{\em timeliness} conditions in prior work~\cite{mayfly}.  Recent work on input
timeliness requires an input value to be used within a
programmer-specified ``expiration'' window after collection~\cite{tics,reliable-time,tardis}.  
These approaches {\em add hardware} to keep track of time during power
failures.  On use of an expired value, the program must recollect the value or
treat the use as an exceptional error case and run mitigation code. 

While prior work has made progress toward the goal of timely intermittent
execution, fundamental challenges remain unaddressed.  
First, the notion of timeliness (which we
call ``freshness'') ignores important cases in which {\em two} input
values must be from the same moment in time, but have
no absolute expiration constraint. We call this timing property 
``temporal consistency'', drawing 
inspiration from data-centric concurrency control~\cite{data-centric-ceze, data-coloring,vaziri-atomic-set-serializability}.  
Temporal consistency ensures that multiple values (e.g., the pressure and
humidity readings) come from the same point in time.  

Second, prior techniques burden the programmer by requiring reasoning about
real time and demanding a distinct expiration time for each value.  If the
programmer incorrectly assigns expiration times, an execution may misbehave
without an expiration time violation. While identifying the data that 
require an expiration time may be simple, assigning the right expiration
time requires choosing the correct real time value for a given program, 
platform, and deployment, which is not simple.  
Some systems~\cite{tics,mayfly} demand more of the programmer, asking for a recovery
action for expired data.

Third, prior timeliness techniques add extra time-keeping hardware: a low-power
real-time clock~\cite{tics,mayfly} or a time-keeper based on charge
remanence~\cite{tardis,reliable-time,forgetfailure}.   The need for time-keeping
hardware precludes the adoption of these techniques on unmodified platforms. 

Fourth, and \emph{most critically}, prior approaches do not formally
define the timeliness properties they aim to provide, nor do they relate the behavior of
an intermittent execution to that of a continuous execution.  Lacking
formal definitions and correctness relations makes it difficult or
impossible to reason if a system is correct.  
A key
contribution of this work is to formally define correctness criteria in 
 relation to continuously-powered executions and
to use these definitions to develop a formalism to prove if a system is correct.

\section{\sys: Correct Inputs via Atomicity}

\sys is a compiler analysis that
inserts atomic regions into code to
enforce freshness and temporal
consistency in intermittent executions
of Rust programs.  \sys is the first
system designed to support the
development of software for
intermittently operating systems using 
Rust.

\subsection{Continuous Execution as a Correctness Spec}
\sys's correctness definitions use the idea that a continuous
program execution is implicitly a
specification of behavior that should
be allowed by an intermittent
execution, including freshness and
temporal consistency properties.
The arbitrary time that passes during a power failure can cause an intermittent
execution to operate on inputs with timing impossible in any continuous
execution, leading to incorrect behavior. 
Prior work~\cite{alpaca,chain,coati,ink} uses \emph{atomicity} of a code region
to keep memory consistent.  We show that atomicity is also
linked to freshness and temporal consistency.

An \emph{atomic region} saves memory state at its start.  If power fails during
a region, the region restores the saved state and execution continues from the
start of the region on reboot. 
A partially executed region's updates to state are not visible to an execution.
If a region completes, its effects become visible to later operations and the
region must have executed without a power failure. If a region executes 
without a power failure, i.e., atomically, its span of code will match the timings of a continuous 
execution.
If multiple input operations execute atomically, the operations
are temporally consistent.
If an input operation and user of the input value execute atomically, 
the value will be fresh when used.
\sys leverages this observation and uses atomic regions as the mechanism to enforce
 time constraints in intermittent executions. 
 Code with freshness or temporal consistency requirements executes in
an atomic region; atomicity ensures that the execution behavior
will match some continuous execution. 
Code with no such requirements executes with JIT checkpoints, enjoying  
the low overheads of taking action only when power fails.

\paragraphb{Jit + Atomic Execution Model}
\label{sec:jit-atomics}
\sys combines JIT checkpoints with atomic regions as an execution model,
 modularly working with any JIT checkpoint and atomic region implementation.
The JIT checkpoint mechanism must checkpoint
volatile memory and registers when energy is low, restoring from that checkpoint on reboot. 
The atomic region implementation must disable JIT checkpoints at the region's entry,  
instead checkpointing volatile system state 
and sufficient non-volatile state to ensure idempotent re-execution of the region~\cite{dino,ratchet,alpaca}.
\sys allows nested or overlapping regions,
flattening them into a single region with the extents of the outermost region.
We describe the implementation of the JIT and Atomic runtimes used in the
evaluation in Section~\ref{sec:runtime-impl} and show the small-step semantics
 in Appendix~\ref{app:semantics}.

\subsection{From Annotations to Correct Executables}
Relying on simple programmer-provided annotations, \sys infers atomic regions that automatically enforce a program's freshness and temporal consistency constraints.
Figure~\ref{fig:sys} illustrates \sys's workflow. 
The
programmer annotates (in blue, upper-left) which variables have freshness or consistency
constraints.  Section~\ref{sec:types} defines the precise meaning of 
these annotations.
\sys must ask programmers for annotations as freshness and temporal consistency
requirements are highly application- and deployment-dependent.
Consider the program in Figure~\ref{fig:problem}. The code logs a pair of values representing sensed pressure
and humidity at line 5.  If power fails between executing lines 4 and 5, the
values are consistent but not fresh. If power fails
between executing lines 3 and 4, the values are neither consistent nor
fresh.  The key challenge is that it is {\em implicit} 
which of temporal consistency and freshness matter for such a pair of values.

\begin{figure}[htb]
    \centering
    \includegraphics[width=1.0\columnwidth]{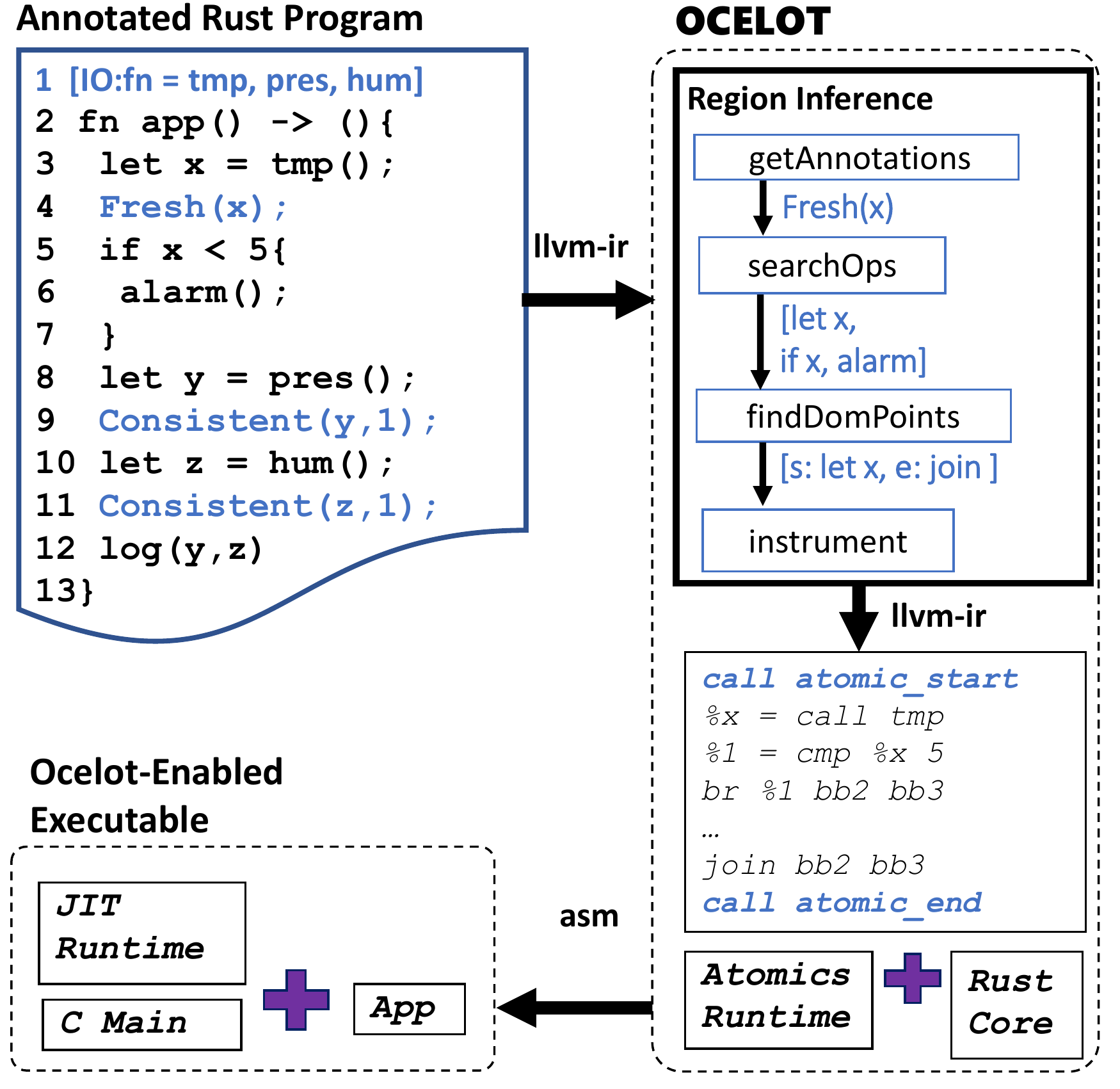}
    \caption{Visualization of the Ocelot toolchain. 
    \label{fig:sys}
    }
  \end{figure}

Ocelot uses the annotations to analyze the code and infer atomic regions that ensure
the constraints.
\sys's region inference algorithm 
searches for operations that 
must execute atomically to enforce each annotation. These operations include inputs 
that each operation depends on and each operation's uses 
of annotated 
data. The algorithm computes points that dominate all such operations, and adds 
a region enclosing those points. Section~\ref{sec:provability} describes 
the algorithm and its correctness; Section~\ref{sec:implementation} gives the implementation details.
Ocelot's compiler links the transformed code to its 
JIT checkpointing and atomic region runtime library (and application libraries), generating a correct executable.

\subsection{Benefit of Targeting Rust} 
To our knowledge, Ocelot is the first
intermittent computing toolchain to target Rust. Enabling correct
intermittent execution of Rust programs is valuable to the
community.
Further, Rust provides memory safety, which contributes to correct
intermittent execution in the following two ways. First, as energy-harvesting
devices are often deployed to inaccessible or remote 
environments, a memory-unsafe program that corrupts non-volatile memory may be difficult
or impossible to patch, making the device useless.
Second, current intermittent systems, including \sys, rely
on the \emph{soundness} of static analyses for their correctness guarantees.
These static analysis identify variables to checkpoint~\cite{dino,
  alpaca, ibis,formal-foundations} or where to place checkpoint
bounds~\cite{ratchet}.  
Pointer alias analysis is a hard problem in C. 
  Missing an alias leads to memory corruption if the compiler 
  fails to checkpoint an aliased memory location that must be checkpointed. Rust's 
ownership and immutability properties make alias analysis more
precise~\cite{crust}. 
Sections~\ref{sec:provability} and ~\ref{sec:implementation} describe how this precision benefits \sys's analyses.
Third, combining \sys with emerging formalisms and frameworks for Rust, such as Rustbelt~\cite{rustbelt, rustbelt2} and 
Iris~\cite{iris} creates a path toward fully formally verified intermittent system implementations. 
\section{Formalizing Freshness and Consistency}
\label{sec:types}

We define a simple modeling language and introduce annotations
for freshness and temporal consistency as discussed in
Section~\ref{sec:background}. Then, we define their meaning by
reference to allowed {\em correct} intermittent executions.

\subsection{A Simple Language}
\label{sec:lang}
This language includes accesses to references and arrays.
A program $p$ consists of a set of function declarations. We assume that the program starts at the $\m{main}$ function. 
We show key syntax below---the rest is in Appendix~\ref{app:syntax}.
\[
\begin{array}{llcl}
\textit{function decls} & \fdecls & \bnfdef& \cdot\bnfalt \fdecls, f(\ft{arg})= c; \m{ret}\;e
\\
\textit{commands} & \cmd & \bnfdef &
\instr \bnfalt \ifthen{e}{\cmd_1}{\cmd_2} \bnfalt \cmd_1;\cmd_2\\
&& \bnfalt &  \elet{x}{e}{\cmd}  \\
&& \bnfalt & \elet{x}{f(v)}{\cmd} \bnfalt \elet{x}{\m{IN()}}{\cmd}\\
&& \bnfalt&     \atomic{\aid, \omega}{\cmd} 
\\

\end{array}
\]
Commands include $\m{if}$
statements, sequencing, variables bindings, function calls,
input operations, and atomic regions, which are parameterized with an
ID $\aid$ and set of checkpointed locations $\omega$.
For simplicity, we assume that $\m{let}$ bound variables are 
mutable and their uses obey Rust's type system, which is the case in our benchmarks.
Commands use values $v$, which are numbers, booleans, or
references, and expressions $e$, which are variables,
values, or operations on sub-expressions. 
A command can also be an instruction
$\instr$, which includes assignments to a dereferenced variable and $\m{skip}$. 
We do not have a general
loop construct as bound loops can be unrolled to $\m{if}$ statements.
Unbounded loops do not introduce technical difficulties, but
complicate the presentation. We do not support recursive functions, which 
many intermittent systems disallow. 

The operational semantic rules
 for continuous executions are of the form: $(\tau_1, N_1, S_1, c_1)
 \SeqStepsto{} (\tau_2, N_2, S_2, c_2)$, where $\tau$ is the logic
 time stamp, $N$ is the nonvolatile memory, $S$ is the calling stack,
 and $c$ is the command to be executed. Intermittent executions are
 of the form $(\timestamp, \context, \nvmem, S, \cmd) \Stepsto{}
 (\timestamp', \context', \nvmem', S', \cmd')$, where $\context$
 is the saved execution context.
 Appendix~\ref{app:semantics} details these rules. These intermittent semantics model \sys's runtime.
 Continuously powered execution traces are
 sequences of $\SeqStepsto{}$ transitions and intermittently powered
 execution traces are sequences of $\Stepsto{}$ transitions. The
 difference is that the latter saves and restores context at
 power failure and reboots, as described in Section~\ref{sec:jit-atomics}.

\subsection{Annotations for Freshness and Consistency}

\sys introduces two annotations: $\m{Fresh}$
and $\m{Consistent}(\mt{id})$.
\[
\begin{array}{llcl}
\textit{commands} & \cmd & \bnfdef &
\cdots \bnfalt  \elet{\m{fresh}~x}{e}{\cmd}  \\
&&\bnfalt&  \m{let}~\m{consistent(n)}~x = e~\ft{in}~\cmd 
\end{array}
\]
Here, $ \m{let}\ {\m{fresh}~x}$ and $\m{let}~\m{consistent(n)}~x$ create
immutable variables.
The annotation for Rust code is shown in the first box of Figure~\ref{fig:sys}. 
On Line 4,
$\m{Fresh}(x)$, declares that any input operations $x$ could depend on
and any uses of $x$ must not be interleaved with a power
failure.
The $\m{Fresh}(x)$ annotation is violated if the input on which $x$ depends
executes before a power failure and a use of $x$ executes after that failure. 
The $\m{Consistent}$ annotation specifies temporal consistency.
The annotation associates a group of variables together into a {\em consistent
set}. For any variable in the consistent set dependent on an input operation, 
those input operations must have executed together with  
no interleaving power failures. 
The annotation takes an ID as a parameter. 
All variables annotated as $\m{Consistent}$ with the same ID are in the same
set, such as $y$ and $z$ in Figure~\ref{fig:sys}. Any input
operations that $y$ and $z$ depend on must execute together as if they were in
a continuously powered execution. 

\subsection{The Meaning of Freshness and Consistency}
\label{sec:defs}
Figure~\ref{fig:properties} illustrates 
freshness and
temporal consistency by relating the intermittent and continuously-powered
executions. 
A double arrow is an intermittent execution trace, and a solid single arrow is
a continuous execution trace. 
The vertical lines mark the transitions (steps) at operations. 
A dashed arrow denotes a 
dependence between operations (i.e., control- and data-flow). Each
operation occurs at a logical time $\tau$, which increases with
each instruction executed. 
\begin{figure}[tbhp]
  \centering
  \subfloat[{\bf  Freshness}]{
  \includegraphics[width=1.0\columnwidth]{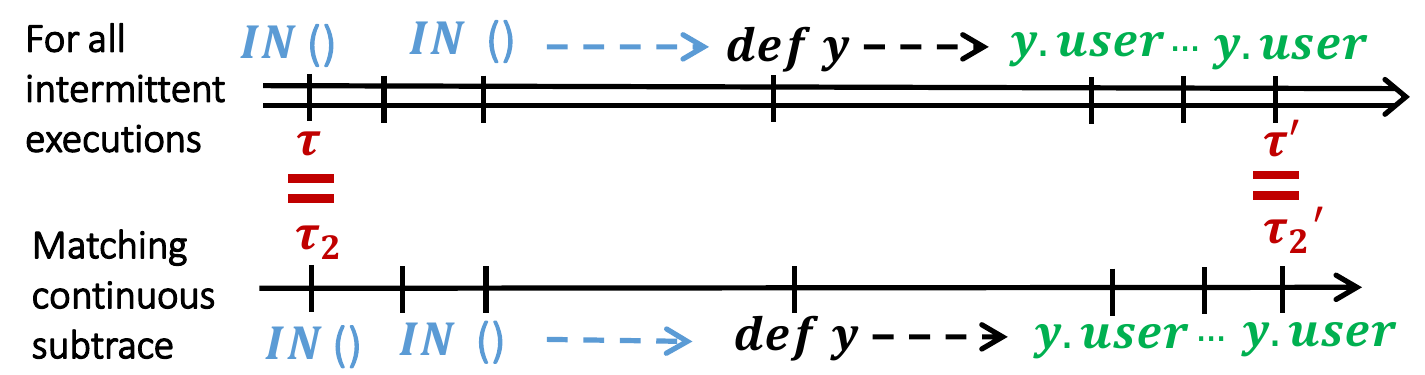}
  }

  \subfloat[{\bf  Temporal Consistency}]{
    \includegraphics[width=1.0\columnwidth]{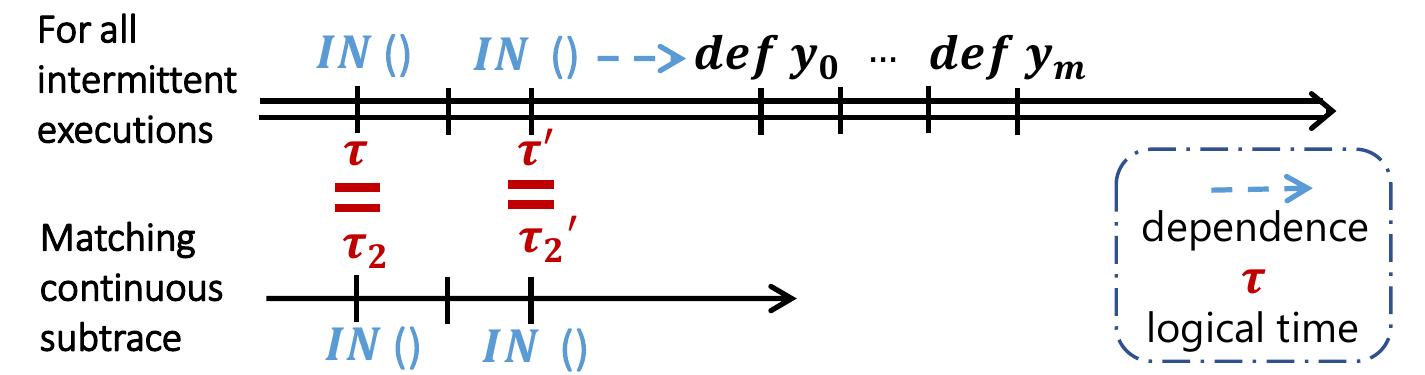}
  }
\caption{Illustrating the properties via execution traces}
  \label{fig:properties}
\end{figure}

The definition of freshness is in 
Figure~\ref{fig:properties}~(a). An intermittent 
system ensures that a variable $y$ is fresh if, for
\emph{all} intermittent execution traces that include
input operations on which $y$ depends (in blue), the definition of
$y$ (in black), and dependents of $y$ (in green), there \emph{exists}
a possible continuous execution of the program that has the same
sequence of operations from the first input to the final dependent
operation.  Furthermore, the time span between the first input to
final dependence on the intermittent execution---marked by the red
timestamps---must match that of this continuous execution.  In this
illustration, a user of $y$ is  
any instruction or command using an expression $e$ where 
one of the terms of $e$ is $y$.
Consider 
a power failure
between $y$'s definition and its first use. 
A JIT checkpointing execution resumes from that point on reboot, but 
after an arbitrary period of time.
The freshness property does not hold: there is no continuous execution with the same operation sequence \emph{and} times. 

The definition of temporal consistency is 
in Figure~\ref{fig:properties}~(b). A set of variables ${y_0 \ldots y_m}$ 
is consistent if, for all intermittent traces with a set
of input operations that $y_i$ depends on (in blue), there exists a
continuous execution with the same sequence of input operations and the
same time difference between the first and last input
operation. 
A power failure between 
input operations violates the property: an arbitrary duration may pass 
during power failure, and no continuous execution
could have the same time difference between operations. 
The definitions of $y_i$ do \emph{not} need to be in
an intermittent subtrace matching a continuous trace for temporal
consistency to hold.

Formal definitions are in Appendix~\ref{app:correctness-defs}.  The key
is to augment the semantics  with taint tracking and store the
input dependency information in memory so we can identify the
input operations on which an annotated variable depends.

\section{Ocelot Design}
\label{sec:provability}

\sys's design generates programs that satisfy freshness
and consistency constraints and we describe how to prove the correctness of our \sys design. 

\subsection{\sys Components}

\sys has two key components to generate programs that are 
 correct-by-construction.  
First, given an annotated program, \sys needs to identify
the instructions that are relevant to each annotation; 
we call this record of an  annotation and relevant instructions 
a \emph{policy}. To construct a policy, \sys must identify the inputs on which an
annotated variable depends, and the uses of any fresh
variable.  \sys constructs a policy using a static taint analysis to track data
and control flow originating at input operations, and builds a taint summary for each
function. Second, given a set of constructed policies, 
\sys adds atomic regions to the program so that all
instructions in a policy are within a single atomic region. To add an atomic region for a policy, \sys
identifies each program point that dominates all instructions in the policy
and inserts the start of an atomic region at those points.  The analysis inserts the end of the atomic region after the last of the instructions in the policy.

\begin{figure*}[t!]
  \centering
\(\small
\begin{array}{llcl@{\quad}llcl}
  \textit{provenance} & \prov & \bnfdef & \m{nil}
  \bnfalt (f_1, \ell_1)::\prov
  &
  \textit{policy} & \pol & \bnfdef &
  \m{fresh}(\ft{decl}:(f, \ell),
  \ft{inputs}: \provs
  \ft{uses}: \lgvec{(f_1, \ell_1)})
  \\
  \textit{policy decls} & \pdecls & \bnfdef& \cdot\bnfalt \pdecls, \pid\mapsto \ft{pol}
  &
  && \bnfalt & \m{consistent}
  (\ft{decls}: \lgvec{(f_1, \ell_1)}, \ft{inputs}:\vec{\prov})
  \\
  \textit{policy map} & \pmap & \bnfdef& \cdot\bnfalt \aid\mapsto \lgvec{\pid}
  &
  \textit{type of taint} & \ft{fromtp} & \bnfdef &
  \m{local}(\ell)\bnfalt \m{retBy}(f, \ell) \bnfalt \m{pbr}(f, \ell) \bnfalt \m{argBy}(\ft{fromtp})
  \\
  \textit{taint map} & \ft{tmap} & \bnfdef &
  \m{ret} \hookleftarrow \ft{inInfo}
   &
  \textit{Inputs} & \ft{inInfo} & \bnfdef &   \emptyset \bnfalt \ft{inInfo}, (\ft{input}:(f,\ell), \ft{fromTp}:\ft{fromtp}) 
  \\
  && \bnfalt &   \m{\&arg} \hookleftarrow \ft{inInfo}
  &   
  \textit{local sum.} & \ft{lSum} & \bnfdef &
  \m{local}(\lgvec{\ft{tmap}})
  \\
  && \bnfalt &  \m{arg} \hookleftarrow \ft{inInfo}
  &
  \textit{caller sum.} & \ft{CSum} & \bnfdef & 
  \m{call}(caller:(f,\ell), \lgvec{tmap})
  \\
  \textit{func sum.} & \ft{fsum} & \bnfdef &\lgvec{lSum}, \lgvec{CSum}
  &
  \textit{func sums.} & \ft{FS} & \bnfdef & \cdot\bnfalt \ft{FS}, f\mapsto \ft{fsum}
\end{array}
\)
\vspace{-1em}
\caption{Syntax for policies and taint maps}
\label{fig:pol-syntax}
\end{figure*}

We formalize policies and summaries of input dependence in
Figure~\ref{fig:pol-syntax}.  We assume that each instruction inside a
function is given a unique label; consequently, a function
name and label pair uniquely identifies an instruction.  To be context
sensitive, we use provenance, the sequence of calls ending in
an input operation, to
distinguish different calls to the same input operation (example to
follow).  A freshness policy is a record containing the declaration, a
list of input operations and their provenance, and a list of uses. A
temporal consistency policy contains a list of declarations and a list
of input operations and their provenance.

The purpose of provenance information is to disambiguate multiple
calls to the same function in a policy. We show an example in
Figure~\ref{fig:impl-examples}~(b).  The main function $\mt{app}$
calls $\mt{confirm}$. $\mt{confirm}$ calls the pressure sensor twice
consistently. 
Both calls to $\mt{pres}$ must occur in the same
atomic region.
To reflect this in the policy declaration, 
each input is associated with its call chain (indicated in
purple) to distinguish the same input with
different calling contexts. 

\begin{figure}[htbp]
  \centering
  \includegraphics[width=1.0\columnwidth]{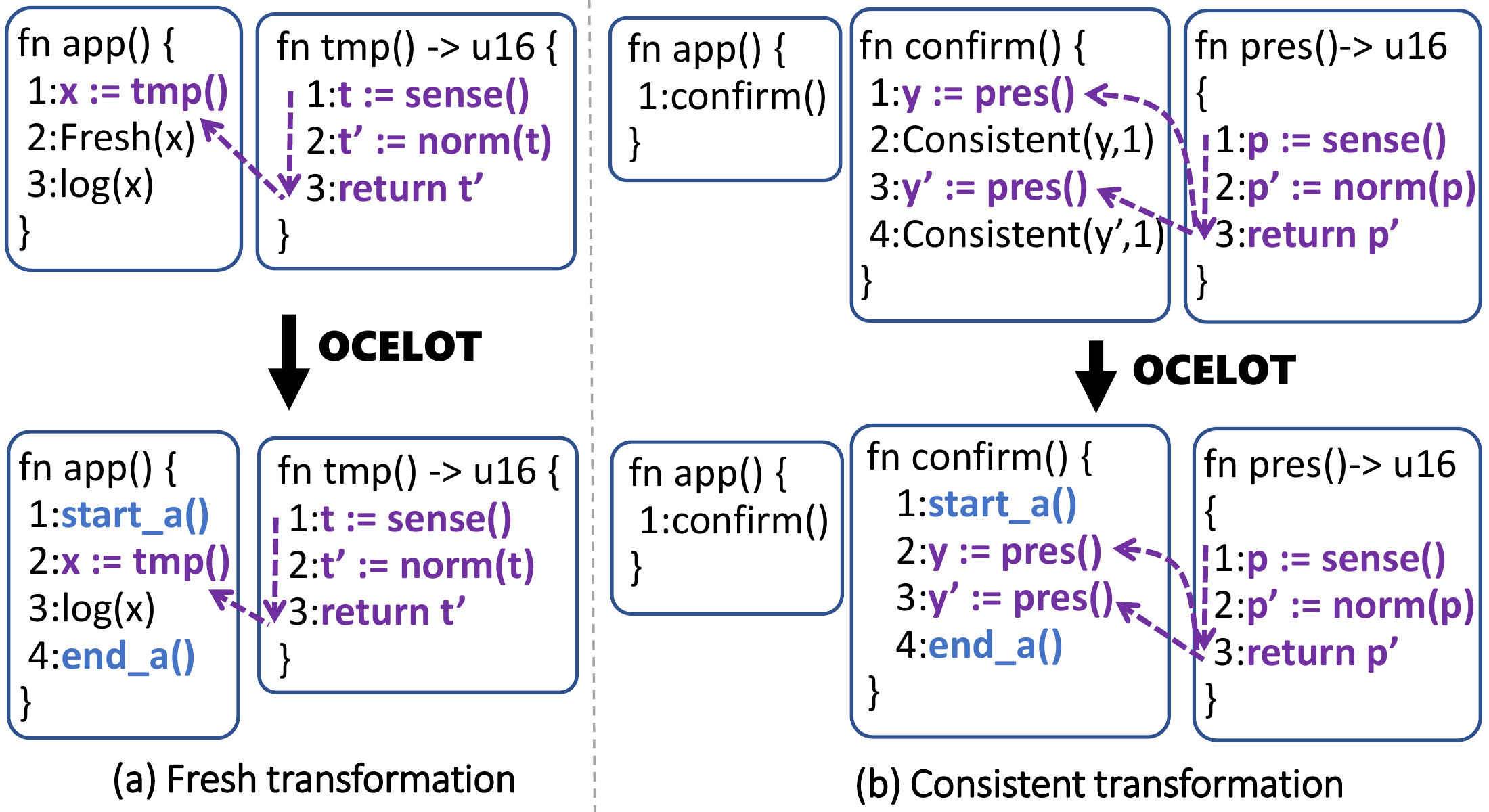}
\caption{Policies for longer call chains}
  \label{fig:impl-examples}
\end{figure}

To present the results of both components, we define policy
declarations \pdecls, which map policy IDs to policies; a policy
map \pmap, which maps atomic region IDs to policies that it enforces;
and function summaries \ft{fsum}, which are
lists of local and caller summaries.
A function summary contains a
taint map entry, which is a link in a call chain describing
 how tainted information flowed into and out of the
function. The entry records if taint flows through the
return ($\m{retBy}(f, \ell)$), into a pass-by-reference parameter ($\m{pbr}(f, \ell)$),
 or is passed in by an
argument ($\m{argBy}(\ft{fromtp})$).  A local summary \ft{lsum} is used if the taint
 was generated within the function, in which case taint
flows to any caller.  For example, input is generated within the
function $\mt{pres}$ and passed through the return,
so $\mt{pres}$ has a local
summary
$\m{local}(\m{ret} \hookleftarrow (\ft{input}{:}\ft{(sense, 0)},
\ft{fromTp}{:}\m{local}(1))$.  A caller summary
\ft{CSum} is used when taint was passed in, in which case taint
flows back only to that calling context. For example, $\mt{norm}$ is called
by $\mt{pres}$ with a tainted argument which flows to the return, so
$\mt{norm}$ has a summary including the taint map
$\m{call}(
caller:(\ft{pres},2),
\m{arg} \hookleftarrow (\ft{input}{:}\ft{(sense, 1)}, \ft{fromTp}{:}\m{local}(1))),
\m{ret} \hookleftarrow (\ft{input}{:}\ft{(sense, 1)}, \ft{fromTp}{:}\m{argBy}(\ft{(pres, 2)}))
)$. Linking taint map entries uncovers 
the entire provenance. 

The two components of \sys are:
{\sc buildSummary}(\fdecls) = (\textit{FS}, \pdecls) and
{\sc inferAtomic}(\fdecls, \ft{FS}, \pdecls) = (\pmap, \fdecls'). We show more 
implementation details in Algorithm~\ref{alg:infer} in Section~\ref{sec:implementation}.

\subsection{Sanity Checks for Results}
\label{sec:check}

Instead of directly proving the algorithms correct, we show a set of
sanity checking rules for the results and prove that programs that
pass these checks can be executed correctly intermittently.  These
rules resemble how the algorithms work and can additionally serve as a
validation tool.

\paragraphb{Checking Summary and Policy Declarations}
We first check that a function summary is correct and that the correct
sets of operations are in policy declarations.  The judgment is of the 
form: $\fdecls;\pdecls,\fsums; (g, \ell); f; M; I \Vdash \cmd : M'; I'$.
$\fsums$ is the summary for all functions.
We are checking the summary for when $f$ is called from $g$ on line $\ell$. 
$M$ and $I$ denote the may-alias and
input-dependence information prior to executing $\cmd$. $M'$ and $I'$
are updated with any may-alias and input dependence information from
$\cmd$.

\begin{mathpar}
  \small
    \mprset{flushleft}
\inferrule[Call-nr]{
    v~\mbox{not a ref.}\\
    \ft{checkUse}(\pdecls, v)\\
    \ft{ins} = I(v) \\
    \fsums(g) = s\\
    \ft{ins} \subseteq s(\m{call}, f, \ell,\m{arg})\\
    \ft{outs} = s(\m{local}, \m{ret}) \cup s(\m{call}, f, \ell, \m{ret})\\
    \ft{outs'} = \ft{outs}[\ft{fromTp}\mapsto \m{retBy}(g, \ell)] \\
    \fdecls;\pdecls,\fsums; c; f; M; I\cup(x \hookleftarrow \ft{outs'})
    \Vdash \cmd: M'; I'
  }{ \fdecls;\pdecls,\fsums; c; f; M; I
    \Vdash  \ell:\elet{x}{g(v)}{\cmd}: M' \backslash x; I' \backslash x
  }
  %

  \mprset{flushleft}
  \inferrule[Let-fresh]{
    \ft{ins} = I(e) \\
    \ft{callChain}(\fsums,\ft{ins})\subseteq \pdecls(\m{fresh}, f, \ell).\m{ins}
    \\\\
    \fdecls;\pdecls,\fsums; c; f; M\cup(x\mapsto M(e)); I\cup(x \hookleftarrow \ft{ins}) 
    \Vdash \cmd: M'; I'
  }{ \fdecls;\pdecls,\fsums; c; f; M; I
    \Vdash  \ell:\elet{\m{fresh}\;x}{e}{\cmd}: M' \backslash x; I' \backslash x
  }
\end{mathpar}

The rule \rulename{Call-nr} 
shows an example of checking function summaries. When calling $g$ with
an argument $v$ (not a reference), if $v$ depends on inputs,
there must be a caller summary for $g$ that records that $f$
propagates taint to $g$. Furthermore, if $g$ returns tainted
information, either locally-generated or due to $f$, those outputs
must be propagated to $x$ when checking the sub-command. We update the
provenance information in the outputs to reflect the fact that the
taint from $f$'s perspective comes from $g$. Further, the second
premise $\ft{checkUse}(\pdecls, v)$ checks that if $v$ is a use of
fresh policy, it has to be in the policy declaration.

The rule \rulename{Let-fresh} checks the fresh annotation. Any input
provenance that the expression of an annotated variable depends on
must be in policy associated with that annotation.
We use $\ft{callChain}(\fsums,\ft{ins})$ to reconstruct the call chain. 
In Figure~\ref{fig:impl-examples} the policy for the freshness example in
(a) must contain the input operation $\mt{sense}$ and its call chain
through the return into $x$ indicated in purple. The rule to check the
consistent annotation (omitted) is similar.
For our example, the two inputs are
(app, 1):: (confirm, 2) ::(pres, 1)::(sense(), 0)
and  (app, 1)::(confirm, 3) ::(pres, 1)::(sense(), 0), showing two different calls to pres.

To check the entire program, we write $\fdecls;\pdecls,\fsums \vdash
\fsums :\m{ok}$ to mean that all the functions are checked under all
specified calling contexts in the summary $\fsums$.

Finally, propagating input
dependence information is simple in this modeling language as there are
no mutable aliases allowed. The may-alias set for a location is always
a singleton set. We can easily find out whether we are writing to a
reference that is passed from the caller, which is difficult for
C and thus the reason why we use Rust.

\paragraphb{Atomic Region Checking}
This check is to make sure that \emph{all} the instructions and their call chains
mentioned in the policy declaration \emph{only}  appear in the correct
atomic region.  We write $\fdecls;\pdecls,\pmap; f; \prov; {\pol}s;
\aid \Vdash \cmd : {\pol}s'$ to mean that command $\cmd$ in function
$f$ is currently called from the call chain $\prov$, within atomic
region $\aid$. {\pol}s are the polices that $\aid$ enforces. After
$\cmd$ is checked, instructions in {\pol}s' still need to appear in
this atomic region. When the $c$ is not in an atomic region, 
{\pol}s and \aid are empty and the end of the judgment is $:\m{ok}$. 
These rules follow each call chain. For each
instruction, the rule checks whether the call chain and
instruction is mentioned in \pdecls. If so, the current
atomic region ID must match that in the \pmap. Then, this
instruction is marked as reached.  At the end of an atomic
region, the rule checks that all instructions in {\pol}s
are reached. Key rules are shown in Appendix~\ref{app:atomic-checking}.
For a program consisting of function declarations \fdecls, 
we say it passes the atomic region check 
if  $\fdecls;\pdecls,\pmap; \m{main}; \cdot; \emptyset;
\cdot \Vdash \fdecls(\m{main}) : \m{ok}$.

\subsection{Correctness}
\label{sec:reg-correctness}

We prove the following correctness theorem. 

\begin{thm}
  Given a program $p$ consisting of functions in $\fdecls$, 
  $\fdecls;\pdecls,\fsums \Vdash \fsums :\m{ok}$ and
  $\fdecls;\pdecls,\pmap; \m{main}; \emptyset;\cdot \Vdash
  \fdecls(\m{main}) : \m{ok}$, then $p$ satisfies all the policies.
\end{thm}

The proof relates the static checking rules to the execution traces, showing that 
if a program $p$ passes the checks
then all input operations that an annotated fresh variable depends on,
as well as any uses of the variable, will be in the same atomic
region. Any input operations that any item in a consistent set depends
on will be in the same region.  As the committed execution of a
region never experiences a power-failure, the committed execution
always has the same timing-behaviour as a continuous execution for any
operations in the region. Thus, w.r.t. to freshness and temporal
consistency, any intermittent execution of $p$ will preserve input
freshness and temporal consistency.

To prove \sys correct, we only need to prove that \sys's algorithms
produce results that pass those checks. This setup allows us to
integrate seamlessly with prior work on proving memory consistency of
intermittent systems~\cite{formal-foundations}.

 \paragraphb{Correctness of Region Size}  There are many possible region placements 
 that could pass the policy check---trivially, $\atomic{\aid, \omega}{\fdecls(\m{main})}$. Another aspect 
 to correct intermittent execution, however, is that any atomic region must be able 
 to complete with the energy that can be stored in the buffer. 
 Thus, Ocelot must infer the smallest regions that satisfy the checks to increase the likelihood that a program 
 is also correct with respect to energy consumption. If the 
 smallest possible region that guarantees correctness w.r.t. to timing policies is too large to 
 complete, such a program fundamentally cannot run correctly.

\section{\system Implementation}
\label{sec:implementation}
\sys's implementation in LLVM uses the region inference algorithm to transform an annotated program 
\fdecls into a program \fdecls' that 
passes the checks of Section~\ref{sec:check}.
The \sys implementation analyzes 
LLVM intermediate representation code generated from an annotated Rust program, determines the policy 
for each annotation, and infers and inserts atomic regions satisfying the policies.
\sys then links with the JIT checkpointing and undo-logging atomic
region runtimes.

\subsection{Mapping Annotations to Policies}
\label{sec:alg-ops}
The implementation of the policy building component closely matches the checking rules in~\ref{sec:check}, except 
that instead of checking that an operation is in the policy declaration, as in rule \rulename{Let-fresh}, 
the algorithm starts with empty policy declarations and adds the operations to the policies at those points. 
The algorithm first finds all annotation instructions, 
which are implemented as calls to the empty functions $\mt{Fresh}\alg{(var)}$ and
$\mt{Consistent}\alg{(var, id)}$. 
The algorithm builds a taint map associating variable
definitions with inputs and the provenance of the input.
Appendix~\ref{app:taint-tracking} shows the
map-building algorithm, which uses a taint tracking analysis that
is inter-procedural, context-sensitive, and leverages Rust's ownership model to
simplify pointer aliasing. We also assume no mutable globals, which are unsafe in Rust. Using 
the input-dependence map, the algorithm adds provenance information 
to the policies as described in Section~\ref{sec:check}. 
After computing the policies, the pass
erases the annotations and starts region
inference.

\subsection{Inferring Atomic Regions}
Algorithm~\ref{alg:infer} performs region inference. Given the function summary and policy declarations 
generated at lines 2 and 3, it aims to generate regions that pass the policy check. 
The algorithm calculates a point that dominates all operations in the policy to begin
the region and a point that post-dominates operations to end the
region. The main challenge is that the policy operations may not be in the same function scope.
The algorithm first finds a candidate function where all operations are either in the function or 
in a descendant of the function. 
It then associates each policy operation with the point 
in the candidate function that reaches the operation. 

To find the candidate, the algorithm maps each policy operation to its basic block 
(Line 5) and calls \rulename{findCandidate} with the block map and the root of the program. 
The function is recursive and tracks which basic blocks in the map execute in 
successor functions from the root.
If all blocks in the map are executed in the current
 root or its successors and no candidate is set, then the 
root returns itself as candidate. Consider example (b) in 
Figure~\ref{fig:impl-examples}. \rulename{findCandidate} starts from $\mt{app}$ and 
calls itself on $\mt{confirm}$. The invocation on 
$\mt{confirm}$ marks that it contains some blocks and calls itself on 
the calls to $\mt{pres}.$ These return the blocks that they called, but no candidate, 
as neither call to $\mt{pres}$ contains all blocks. Combining the results 
of its successors, $\mt{confirm}$ does contain all blocks. The invocation marks $\mt{confirm}$
as the candidate function, returning this to the invocation on $\mt{app}$. While $\mt{app}$ is also 
a root of all the blocks, the candidate is already set, so the invocation returns $\mt{confirm}$. 
Placing the region in $\mt{confirm}$ results in a smaller
 region than placing it in $\mt{app}$. 

 \begin{algorithmic}[1]
  \Function{inferAtomic}{\alg{Cmd}}
  \State \alg{map \gets buildSummary(Cmd)}
  \State \alg{pol \gets buildPolicies(Cmd, map)}
   
  \ForAll{\alg{set \in pol}}
  \State \alg{\forall item \in set, blocks[item] \gets item.basicBlock}
  

  \State \alg{goalFunc \gets findCandidate(blocks, Cmd.root)}
  \ForAll{\alg{item \in set}}
  \While{\alg{blocks[item].func \neq goalFunc}} 
    \State \alg{calls \gets blocks[item].func.callers()}
     \ForAll{\alg{ call \in calls}} 
    \If{\alg{call \in set}} 
    \State \alg{blocks[item] \gets call.basicBlock}
    \EndIf
    \EndFor
   
  \EndWhile
  \EndFor
  \State \alg{startDom \gets closestCommonDom(blocks)}
  \State \alg{endDom \gets closestCommonPostDom(blocks)}

  \State \alg{(S, E) \gets truncate(startDom, endDom, set)}
  \State \alg{Cmd.insertRegAt(S, E)}
  \EndFor
  \EndFunction
  \captionof{algorithm}{Atomic Region Inference}
  \label{alg:infer}
  \end{algorithmic}

To find the points in the goal
function that reach a policy operation, the algorithm traverses the call graph aided with the
basic block map (lines 8-15).
Until the function of each basic block in the map is the goal function,
the algorithm gets
the callers of the function and checks if the callsite is in the
policy, as the policy includes the provenance. 
If it is, traversing this path will get the
basic block closer to the goal function. The algorithm sets the map value
to the basic block of the callsite. For the 
freshness example in Figure~\ref{fig:impl-examples}, the basic block of the assignments to $t,t'$ 
is in the function $\mt{tmp}$. $\mt{tmp}$ is called by $\mt{app}$
at the callsite $x := \mt{tmp}$. This operation is in the policy, 
so the map values for $t, t'$ are set to the basic block of the 
callsite. Now all blocks in the map are in $\mt{app}$.

Once all blocks associated to the policy operations are in the same function, the
algorithm can use LLVM's built-in \rulename{closestCommonDominator}
and \rulename{closestCommonPostDominator} passes, returning
candidate \alg{startDom} and \alg{endDom} basic blocks (lines
17-18). Multiple returns in the source function do not cause the
post-dominance analysis to break, as the compiled code has a
return landing-pad that post-dominates all paths through the function.
Taking these blocks, the algorithm calls \rulename{Truncate}, 
which finds the latest point in the starting block
that dominates everything in the set and the earliest point in
the ending block that post-dominates everything in the
set. Inserting region start and end instructions at these points
creates an atomic region containing all the operations in the policy.

\subsection{Runtime Implementation}
\label{sec:runtime-impl}
To implement atomic regions with undo logging, we used WAR and EMW
analysis code publicly available from prior work~\cite{alpaca, formal-foundations}, porting both to
work for Rust code. The existing
implementation has a \alg{currentContext} variable that tracks region metadata.
We add to it a \alg{mode} field that is either \alg{jit}
or \alg{atomic}. The value is \alg{atomic} in an atomic region,
and is \alg{jit} otherwise.  An atomic region's checkpoint also saves volatile
execution context (registers, stack) along with performing undo-logging.  
The routines to save and restore volatile execution
context are the same for both JIT checkpoints and atomic regions, and 
 are similar to Hibernus~\cite{hibernus}. The checkpoint
routine copies registers and stack to a dedicated non-volatile memory region. 
Restoration copies values from non-volatile memory back into the context.
  
We target the Capybara energy-harvesting hardware platform~\cite{capybara}, which has a built-in comparator to 
monitor energy, the only hardware needed for JIT checkpointing. 
The firmware triggers an interrupt on low energy. 
We raised the voltage level on which the interrupt triggers and modified the ISR 
to handle JIT mode and atomic mode. 
In JIT mode, the ISR checkpoints volatile state and shuts down. 
In atomic mode, the ISR only shuts down. 
Similarly to Samoyed~\cite{samoyed}, we assume that the extra energy gained from raising the trigger point 
will always be enough to 
complete the checkpoint. As pointed out in prior work~\cite{tics,chinchilla,samoyed}, this assumption 
may not be true for programs with large and unpredictable stack sizes. 
None of our benchmarks have this behaviour and our 
implementation is sufficient to demonstrate \sys's correctness improvements with low overhead. 

\section{Evaluation}
\label{sec:evaluation}

We evaluate the performance and correctness of programs generated by \sys and the 
programmer effort of using \sys. 
We measure runtime overhead of a set 
of benchmarks compiled with \system, with just JIT checkpoints, and with just Atomic 
regions (similar to the execution model of DINO~\cite{dino}).
We measure the runtimes 
on continuous power, showing the inherent performance overheads of \system and Atomics even 
when energy is plentiful, and on intermittent power. While JIT is fastest, it is incorrect.  
\system has a mean 7\% runtime increase and is correct by construction.
To show correctness empirically, we run the \system programs with simulated power failure points 
chosen to be sufficient
to uncover any timing violations 
and on real intermittent power.
Finally, we compare the code changes needed to write correct 
programs with Ocelot, TICS~\cite{tics}, and Samoyed~\cite{samoyed}. We further discuss the 
difference between annotating code and manually adding atomic regions in Section~\ref{sec:discussion}.

\subsection{Benchmarks}
We use the following 6 benchmarks that are representative
 of sensor applications in the intermittent computing domain. 
 \alg{ Activity}, an activity recognition app, \alg{ Greenhouse}, a greenhouse monitor app, 
 \alg{ CEM}, a compression logger, \alg{ Photo}, an app that takes the average of 5
 photo-resistor readings, \alg{ Send Photo}, an app that samples a photo resistor and sends 
 a radio packet if the value is too high, and \alg{ Tire}, a tire pressure monitor.
 All benchmarks except for  \alg{ Tire} were originally written in C, ported to Rust by us.
 \alg{ Activity} and \alg{ Greenhouse} were obtained from the TICS artifact~\cite{tics}, 
 \alg{ Photo} and \alg{ Send Photo} were microbenchmarks used in Samoyed~\cite{samoyed} and were obtained from 
 the authors, and \alg{ CEM} is originally from DINO~\cite{dino}. 
 \alg{ Tire} we wrote ourselves.
 We characterize the benchmarks in Table~\ref{tab:bench}. The table shows 
 the provenance of each benchmark, the lines of code, the sensors used 
 or simulated (denoted with an asterisk), and the constraints 
 used. Comma-separated values mean that the constraints apply to separate pieces of data. 
 "FreshCon" means that both constraints were used on the \emph{same} piece 
 of data. Both unaltered and annotated benchmarks are located at 
 \url{https://github.com/CMUAbstract/ocelot}.

 \begin{table}[tbp]
    \centering
    
    \scriptsize
    \setlength\tabcolsep{1.2pt}
    \begin{tabular}{c|l|rrr}
    \hline  
    {\bf Origin}&{\bf App} & {\bf LoC} & {\bf Sensors} & {\bf Constraints}\\ 
    \hline
    \multirow{2}{*}{{\bf TICS}} 
    &  {\bf Activity} & 470 & Accel* & Con, Fresh \\
    &{\bf Greenhouse} & 170& Hum, Temp & Con \\
    
    \hline
    
    \multirow{2}{*}{{\bf Samoyed}} 
     &{\bf Photo}& 68& Photo& Con\\
     &{\bf Send Photo}& 92& Photo& Fresh\\
    \hline

    {\bf DINO}&{\bf CEM}& 292& Temp*& Fresh\\
    \hline
    {\bf \system}&{\bf Tire}& 338& Pres*, Temp*, Accel*& Fresh, Con, FreshCon
    \end{tabular}
    \caption{Benchmark Characteristics. 
    The origins:~\cite{tics, samoyed, dino}}
    \label{tab:bench}
\end{table}

\subsection{Overheads}
The goal of the performance evaluation is to make a generalizable comparison of \sys, which 
uses a JIT + Atomics execution model, to systems that use just Atomics~\cite{dino, coati, ratchet, chinchilla, mayfly,alpaca} or just 
JIT~\cite{hibernus,hibernusplusplus,idetic,quickrecall}.
We ran the benchmarks on the Capybara hardware platform~\cite{capybara}, harvesting 
energy from a PowerCast~\cite{powercast} 
antenna placed 10 inches away. 

\begin{figure}[htb]
    \centering
    \includegraphics[width=1.0\columnwidth]{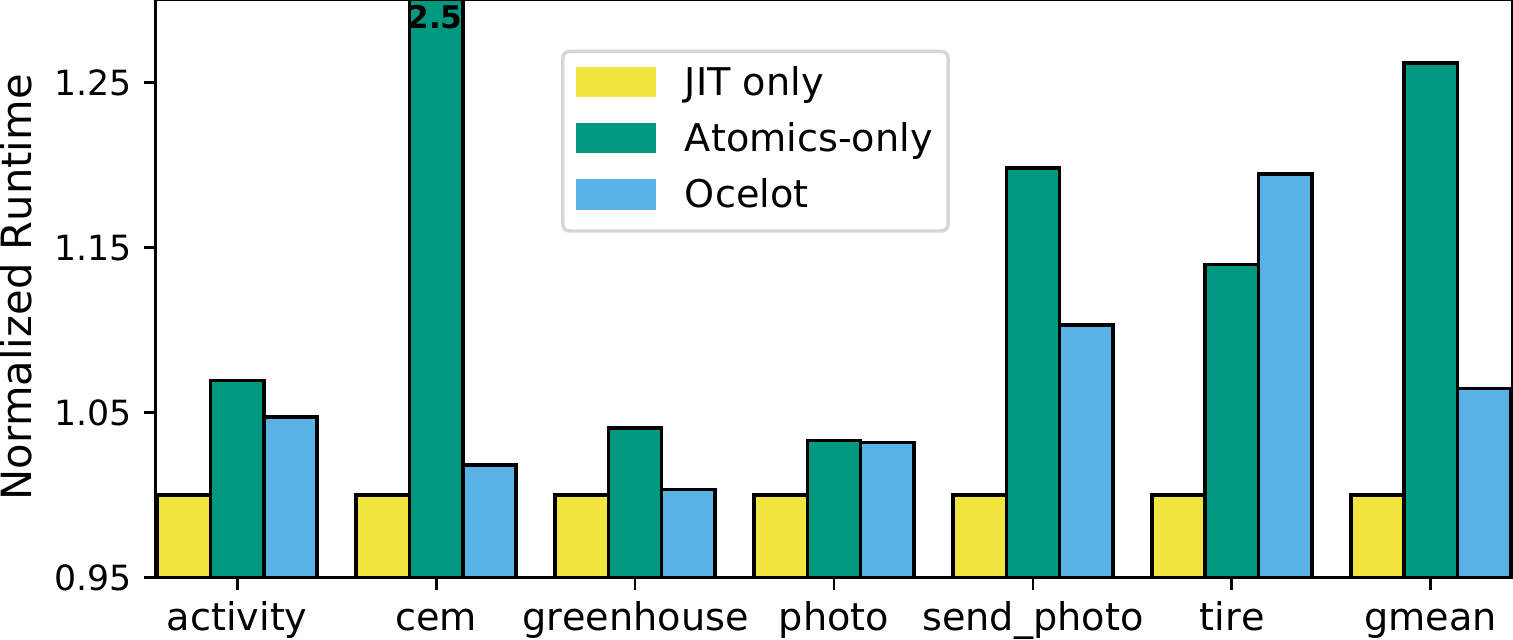}
    \caption{Continuous runtimes of JIT, Atomics, and \system}
    \label{fig:cont-runtimes}
\end{figure}
Figure~\ref{fig:cont-runtimes} shows the runtimes on continuous power 
 of the benchmarks compiled 
with JIT checkpoints only (yellow columns), with Ocelot-inferred Atomic regions 
(blue columns), and 
with Atomic regions only (teal columns). To enable 
correct output, calls to the UART were guarded by a small atomic region, generating 
a constant overhead for all configurations. 
Runtimes are normalized to the JIT execution, which has the least overhead at 
the cost of correctness, both of timing constraints and of basic peripheral operation~\cite{samoyed, sytare, restop, karma}.
 The y-axis shows 
the runtime increase, and the x-axis shows the benchmarks. The Atomics-only programs are entirely 
divided into atomic regions. We manually placed regions where Ocelot-inferred regions would go, to ensure that the 
correctness properties will still be satisfied, which is otherwise not guaranteed. 
If statically-placed checkpoints or tasks were used on the program in a prior work 
(\alg{Greenhouse}, \alg{Activity}, and \alg{CEM}), we tried to place atomic 
regions as similarly as possible. \alg{CEM} required a few code changes to run on 
the device, as the original program had a region with a WAR dependence on a large structure. 
Backing the entire structure to the undo log caused the program to be too large to flash 
to the device. We changed the code to remove any WAR dependences on that structure. 
Generally, atomic regions, whether manually placed or inferred add a reasonable 
amount of run 
time overhead. The geometric mean runtime increase of Ocelot programs to JIT  
is around 7\%. Atomics-only experiences similar overheads, except for CEM which 
has a 2.5 runtime increase. CEM grabs a sensor value once and then performs 
lookup and insertion into a compressed log. The inferred atomic region is small and 
infrequently executed, resulting in an Ocelot runtime that is close to JIT. With Atomics-only, 
all lookup and insertion code is in regions even though re-execution is unnecessary for either timing 
or memory correctness, 
resulting in a large overhead. \alg{Tire}, in contrast, is slightly faster with Atomics-only 
than with Ocelot. The Atomics-only version nests a frequently executing 
inferred region within a larger, less frequently executing region. At runtime, only the outermost bounds 
are treated as an atomic region.

Next, we show the runtimes of the benchmarks on intermittent power in 
Figure~\ref{fig:int-runtimes}. All bars are normalized to the JIT execution time 
on continuous power. Again, yellow represents 
JIT, blue represents Ocelot, and teal represents Atomics-only. For each benchmark, 
the lower, colored bar represents the time spent running the application, and the stacked 
grey bar represents the time spent off, charging. The lower sublot shows 
a closer view of the time spent running the application. Since JIT cannot execute peripheral 
operations correctly~\cite{samoyed, sytare,restop, karma}, 
we changed \alg{Greenhouse}, \alg{Photo}, and \alg{Send-Photo} 
to simulate sensors . Generally, the intermittent overheads have the 
same proportion as the continuous ones. A notable difference between the plots 
is that the runtime is dominated by charging time. The benchmarks were run on 
real hardware and harvested energy; the off, charging times are dictated by the 
physical environment.

\begin{figure}[htb]
    \centering
    \includegraphics[width=1.0\columnwidth]{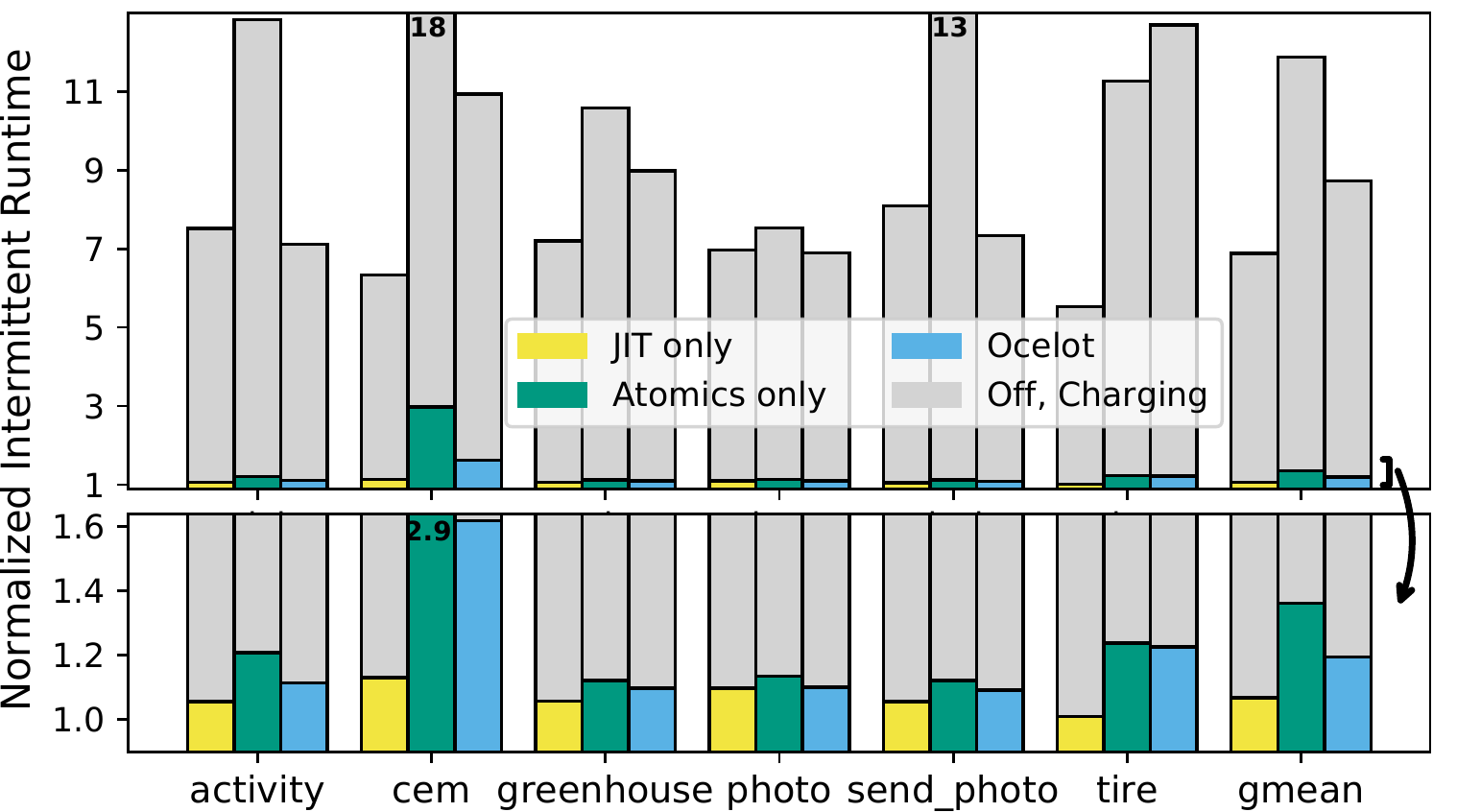}
    \caption{Intermittent runtimes of \system, JIT, and Atomics}
    \label{fig:int-runtimes}
\end{figure}

\subsection{Correctness}
We showed how to check Ocelot's correctness on a modeled language in Section~\ref{sec:provability}.
Here we empirically show the correctness of the implementation. 
Power can potentially fail at any 
instruction in an intermittent execution. To determine if a program will violate freshness and consistency 
policies, however, we must consider power failures only at a few 
key locations; 
there 
must not be a visible power-failure between the 
inputs and the dependencies of a fresh variable, and there must 
 not be a visible power failure between 
the inputs of a consistent set. Power failures outside these sub-traces
do not affect if the policy is 
upheld. We insert simulated 
power failures immediately before the use of a fresh variable and between input operations 
in a consistent set.
Power-failing at each instruction 
is unnecessary, as these failure points are sufficient to expose if 
the atomic region is placed correctly and will re-execute all necessary inputs.

To determine if an input is gathered before 
a power failure, we add bit vector in nonvolatile memory. Each 
sensor operation has a unique position in the bit vector.
On an input operation, the sensor's position in the bit vector is set to 1. 
On power failure, the bit vector is cleared. 
On the use of a fresh variable, the bits of any dependent sensors are checked. 
On an input operation in a consistent set, the bits of any preceding 
operations in the set are checked.
If the sensor has not been re-executed, the checked 
bit will be zero, generating an error. 
Table~\ref{tab:sim-correctness}~(a) 
shows the results of injecting these simulated power failures. \system programs 
did not experience any violations, whereas JIT programs always did.
\begin{table}
    \centering
    
    \scriptsize
    \setlength\tabcolsep{1.2pt}
\subfloat[{\bf Violating \% with pathological power failure points}]{

    \begin{tabular}{l|rrrrrr}
    \hline  
     &\multicolumn{6}{c}{{\bf Percentage Violating}}\\
     {\bf Exec. Model}&{\bf Activity} & {\bf CEM} & {\bf Greenhouse} & {\bf Photo} & {\bf Send Photo} & {\bf Tire}\\ 
    \hline
    {\bf \system}& 0\%&  0\%&  0\%&  0\%&  0\%&  0\%\\
    {\bf JIT}& 100\%& 100\%& 100\%& 100\%& 100\%& 100\%\\
    \end{tabular}
}

\subfloat[{\bf  Violating \% while running intermittently}]{
    
    \begin{tabular}{l|rrrrrr}
        \hline
        {\bf Exec. Model}&{\bf Activity} & {\bf CEM} & {\bf Greenhouse} & {\bf Photo} & {\bf Send Photo} & {\bf Tire}\\ 
        \hline
    {\bf \system}&0\%& 0\%& 0\%& 0\% &0\%& 0\%\\
    {\bf JIT}&50\%& 0\%& 24\%& 77\% & 50\%& 3\%\\
    \end{tabular}
}
\caption{Correctness comparison of \system to JIT}
 \label{tab:sim-correctness}
\end{table}

\begin{table*}[!t]
    \centering 
    \scriptsize
    \begin{tabular}{c|l|l|l|l}
        \hline
        {\bf System} & {\bf Constructs} & {\bf Strategy} & {\bf LoC Changes} &
        {\bf Correctly Upholds Freshness and Consistency}\\

        \toprule
        {\bf Ocelot} & Time-constraint Types& Annotate inputs, & 
        1*(num inputs) &
        \textcolor{teal}{\bf Correct. Intermittent execution}\\
        & &   time-constrained data  & + 1*(data with constraint)  & 
        \textcolor{teal}{\bf must match the continuous specification}\\
        \toprule
        {\bf JIT} & None& Do nothing & 
        0 &
        \textcolor{red}{\bf Incorrect}\\
        \hline
        {\bf Atomics} & Atomic Regions& Annotate inputs, & 
        1*(num inputs)   &
        \textcolor{orange}{\bf Programmer-dependent}\\
        & &  manually place regions.  & + 2*(num atomic regions) & 
        \textcolor{orange}{\bf could place regions incorrectly} \\
        & &  Reason about control, data flow.   &  &  \\

        \hline
        {\bf TICS} & Timestamp alignment, 
        & Add real-time expiry date, & 
        3*(time-sensitive data)  &
        \textcolor{teal}{\bf Real-time timeliness}\\
        &Expiration Catch, &  timestamp alignment operations,  
        & + $\Sigma^n_{i=0}(\mt{LoC~of~handler_i})$ & 
        \textcolor{orange}{\bf No clear mapping to temporal consistency} \\
        & Timely Branches &  and 
        expiration/branch points.  &  & \\
        & &  Write exception handlers.  &  & \\

        \hline
        {\bf Samoyed} & Atomic Functions& Reason about control, data flow.&
        $\Sigma^n_{i=0}(\mt{rewrite~cost~of~}f_i)$ + &

        \textcolor{orange}{\bf Programmer-dependent}\\
        & & Rewrite code to be function,   & $\Sigma^n_{i=0}(\mt{LoC~of~scalingRule_i})$ + & 
        \textcolor{orange}{\bf could put wrong code in atomic function} \\
        & & (opt) provide software fallback,  &  
        $\Sigma^n_{i=0}(\mt{LoC~of~fallback_i})$ &  \\
        & &  (opt) scaling rules.  &  & \\

    \end{tabular}
    \caption{Characterizing the Strategy of Using Ocelot}
    \label{tab:usability}
\end{table*}

The previous experiment shows 
that a policy violation cannot 
 occur on \system programs. To show that violations do 
 occur practically as well as theoretically on JIT programs, 
 we ran the programs with the added bit vector on intermittent power, 
 using the simulated sensor versions of \alg{ Greenhouse}, \alg{ Photo}, and \alg{SendPhoto}.
  We ran each benchmark 
 for a fixed time of 100 seconds and recorded the percentage of complete runs of the benchmark 
 that contained a policy violation. Each benchmark completed between 50 - 450 
 times, depending on the program runtime. The results are in Table~\ref{tab:sim-correctness}~(b).
 All benchmarks except 
 \alg{ CEM} experienced a policy violation within that window. \alg{ CEM} is 
 a compute heavy benchmark, and the freshness constraint only applies for 
 a few instructions. A policy violation is possible, but 
 experiencing a power failure at exactly the right 
 point is rare. Benchmarks like \alg{ Activity}, \alg{ Photo}, and 
 \alg{ Send Photo} have time constraints that cover much of 
 the program, so violations are frequent. 

 \subsection{Comparing Code Changes}
 \label{sec:prog-effort}
 \todo{Flip some of this to discussion section?}
We characterize the effort of using Ocelot. We compare the JIT baseline, 
Ocelot, and Atomics-only, plus the prior works TICS and Samoyed in Table 
~\ref{tab:usability}. The first column of the table 
shows the system. Column {\bf Constructs} shows 
the language constructs each system provides to the 
programmer to enable correct execution with inputs. Column 
{\bf Strategy} lists in brief the method to 
use the constructs. Column {\bf LoC Changes} estimates 
the lines of code needed to implement the strategy. The last column 
indicates if the methods succeed in providing fresh and temporally consistent 
intermittent executions.

\system 
requires only a small, bounded amount of code changes. The 
programmer must declare which functions generate input and apply 
\fresh and \consistent annotations to variables. Each annotation  
requires adding a single line of code, and the programmer \emph{never} has to 
write new program logic. 
The resultant program is correct by construction. 

JIT checkpointing provides nothing to the programmer, requiring no effort but 
offering no correctness. 
Atomics-only requires the programmer to 
reason about the dataflow and relationships of input operation to each other 
and place the regions. Since undo-logging backs up EMW sets~\cite{formal-foundations}, 
the programmer must also specify inputs. If the programmer reasons correctly, 
the resultant program will be correct.

TICS~\cite{tics} offers the programmer 
annotations that 
require reasoning about real time.
It provides expiry times, a timestamp alignment operator, an expiration check,
and a timely branch check. The latter checks also allow the 
programmer to specify an exception-like handler to execute if the check fails. Handlers impose an unknown 
burden on the programmer as they have to write new logic. If the original program 
has explicit real-time checks and exception handling, 
the process is straightforward and is a good match for TICS. 
Otherwise, the programmer must generate these from scratch.
TICS ensures that stale data is not processed, similar 
to freshness, though it does not guarantee the existence of a 
continuous execution with the same behaviour. 
If the programmer chooses an expiration time 
poorly, the program could behave in undesired ways. 
The TICS concept of timeliness does not cover temporal consistency.   

Samoyed~\cite{samoyed} focuses on safe peripheral operations and provides 
the programmer with atomic functions. Samoyed requires more rewriting work 
than simple atomic regions, as the code to be executed 
atomically must be a function. The programmer can also specify scaling 
rules and fallbacks, if the function takes too much energy to execute within a 
power cycle. 
If the 
programmer carefully reasons about the dependencies and relationships of 
input operations, they can use atomic functions to uphold freshness and consistency.

In Table~\ref{tab:concrete}, we model the concrete lines of code needed 
to enable correct execution on each of our benchmarks for \system, TICS, and Samoyed. For TICS, 
we estimate that each handler will take five lines of code. 
For consistent sets, we estimate that each variable incurs 2 LoC changes (expiry and timestamp alignment), 
but that there is only one expiration check and accompanying handler per set.   For Samoyed, 
we estimate that restructuring into atomic functions will take a fixed 3 LoC 
 (creating the atomic function signature, adding the callsite), plus an additional 
line for each parameter to the function. Scaling rules take 3 LoC, 
fallbacks take 5 LoC, and these are provided for any atomic function with a 
loop. For all benchmarks, \system requires the fewest annotations. Moreover, Ocelot 
does not require reasoning about real-time values, about information flow from 
inputs, or writing exception handling, instead enforcing correctness by construction. 

\begin{table}
    \centering
    
    \scriptsize
    \setlength\tabcolsep{1.2pt}
    \begin{tabular}{c|rrrrrr|l|l}
        & \multicolumn{6}{c}{{\bf LoC Changes}} & 
        
        {\bf Real-time } & {\bf Data-flow }\\ 
        {\bf Sys} &{\bf Act} & {\bf CEM} & {\bf G-house} & { \bf Photo}  
        & {\bf S-Photo} & {\bf Tire} & 
        {\bf Reasoning} & {\bf Reasoning}\\ 
        \hline
        {\bf \system} & 5 & 2 & 7 & 2 & 4 & 9 & \textcolor{teal}{\bf No} & \textcolor{teal}{\bf No}\\ 
        \hline
        {\bf TICS} & 20 & 8 & 12 & 8  
        & 8 & 32 & 
        \textcolor{red}{\bf Yes} & \textcolor{teal}{\bf No}\\ 

        {\bf Samoyed} & 18 & 4 & 6 & 12  
        & 4 & 24 &  
        \textcolor{teal}{\bf No} & \textcolor{red}{\bf Yes}\\ 
        
    \end{tabular}
    \caption{Effort of using \system vs. TICS and Samoyed}
    \label{tab:concrete}
\end{table}

\section{Discussion of Annotation Benefits}
\label{sec:discussion}
In this section we discuss the benefits of 
Ocelot annotations as compared to manually adding atomic regions.
Instead of using Ocelot annotations and allowing the system 
to infer atomic region placement, 
programmers can carefully place atomic region constructs 
to uphold timing constraints, but such an approach has several drawbacks. 

\paragraphb{Annotation Simplicity and Meaning} While adding \sys annotations and manually adding atomic regions both 
require the programmer to be aware of timing invariants in their program, 
programmers must use additional reasoning to 
correctly place atomic regions.  
Figure~\ref{fig:snippet1} shows a code snippet from the tire benchmark. 
The snippet describes the decision whether or not to send out a burst tire alarm. 
This decision should happen on a fresh sensor reading, and variables in the branch should 
be consistent with each other. Such a level of knowledge about program behaviour 
is sufficient to add \sys annotations -- {\tt currMotion} and {\tt avgDiff} should be marked 
$\m{Fresh}$ and $\m{Consistent}$ as in lines 1-2. 
\begin{figure}
    \scriptsize
\begin{verbatim}
1    FreshConsistent(avgDiff, 1);    
2    FreshConsistent(&currMotion,1);
3    if isMoving(&currMotion) && avgDiff > 0 {
4        sendData("urgent_burst_tire!\r\n\0");
5        *urgentWarningCount +=1;
6    }
\end{verbatim}
\caption{Tire code snippet}
\label{fig:snippet1}
\end{figure}

To manually place an atomic region, the programmer has to examine the 
data each of the variables depends on and make sure any 
inputs in that data flow are included in the 
atomic region. The programmer must know the invariants in either case, 
but adding an atomic region that includes every input the variables depends on and 
every use of the variables requires more work than annotating the variables at the declaration point only.
\todo{Make less abstract}
Even knowing the invariants, the programmer could make a mistake when manually adding a region, which 
would not be detected by the system as added atomic regions do not carry any specification information. 
The program has no explicitly declared guarantees of what properties are met. 
When using \sys annotations, however, the programmer is explicitly giving a specification of the 
timing properties that must be upheld, and the \sys-generated program will correctly uphold 
that specification. 
\limin{Can we explain the example fully, show what dependences e.g., y depends on ..., 
but not ... need to be identified? Now the description seems abstract}

\paragraphb{Region Size} As discussed in Section~\ref{sec:reg-correctness}, \sys's 
implementation aims to find the smallest region that satisfies the specified timing 
constraints. A programmer-added region may be uncessarily large. Consider the programming 
pattern in Figure~\ref{fig:snippet2}. 
\begin{figure}
    \scriptsize
\begin{verbatim}
1  fn main ()  {               fn confirm() { 
2    //should be consistent       let y = pres();
3    confirm();                   Consistent(y,1);
4  }                              let y' = pres();
5                                 Consistent(y,1);
6                                 ... //more processing
7                              }
\end{verbatim}
\caption{The intuitive atomic region around \rulename{confirm} could be too expensive}
\label{fig:snippet2}
\end{figure}
The function $\mt{main}$ calls function $\mt{confirm}$ which has a 
temporal consistency constraint on the assignments to $y, y'$. Programs 
with this pattern will likely do more processing on $y,y'$ in $\mt{confirm}$. If a programmer 
manually adding regions knows that $\mt{confirm}$ calls sensors that need to be consistent, they 
may simply wrap the entire function in an atomic region. While such a region placement does preserve the 
timing constraints, it uncessarily includes any processing in $\mt{confirm}$, while the 
\sys region would not. If sampling the sensors \emph{and} processing the data takes more 
energy than can fit in the buffer, 
the program with manually-added regions would fail to complete, while the \sys program would not. 
If an \sys program fails to complete, the specified timing constraints are fundamentally unsatisfiable 
with the energy capacity of the device.

\paragraphb{Using added regions and \sys together} Programmers may have programs that 
already have atomic regions placed, e.g., 
if they used Samoyed~\cite{samoyed} to write programs with safe peripheral 
operations, or otherwise want manual control over atomic region placement (that they 
are sure will run to completion). 
\milijana{vague...} \sys can be used with 
programs that already have atomic regions.  In this use case, \sys's analysis confirms that the region placement 
meets a program's annotated timing constraints. If the input to \sys is a 
program that already has atomic regions as well as annotations, 
\sys adds regions to enforce the annotations. 
While these added regions may overlap or duplicate existing ones, 
only the outermost bounds of nested regions execute (see 
Appendix~\ref{app:semantics}). The resultant program respects the atomicity of 
both programmer-specified and inferred regions without extra runtime overhead. 
Thus, \sys in conjunction with manually added regions can give the programmer control 
and correctness. Additionally, extending \sys with a true checker mode is straightforward.
After generating the policy sets, \sys could merely check that all instructions in 
each set are dominated by existing region boundaries, instead of inferring and placing the region boundaries.
\limin{the above is a description of what we do now. Can you imagine adding a flag to the compiler to say it's for checking only mode, and only calls parts of the algorithm implemented in \sys? This would be a more clean way to use \sys as a checker. }

\section{Related Work}

Areas related to \sys are intermittent systems with timeliness and reactivity,
 work on persistent memory correctness and crash-consistency, 
and data-centric concurrency. 

\noindent\textbf{Intermittent Systems with Inputs}
MayFly~\cite{mayfly} introduced the concept of timeliness, but its solution 
is complicated, requiring programmers to write programs as dataflow graphs with expirations 
on the edges. TICS~\cite{tics} 
is most similar to this work, providing timely intermittent 
computation through annotations on existing programs. In contrast to \sys, 
both these works require reasoning about real-time, do not examine temporal-consistency,  
and 
require additional hardware to keep time through power failures
~\cite{reliable-time, forgetfailure}. 
TICS also presents an architecture for 
constant-time checkpoints, which is complementary and can be used with \sys. 

Samoyed~\cite{samoyed}, RESTOP~\cite{restop}, Sytare~\cite{sytare} and Karma~\cite{karma} 
all address the problem of safe peripheral manipulation 
on intermittent systems, but do not consider application-level time-constraints. 
Samoyed provides atomically executing functions which can be used to ensure 
freshness and consistency, though at more effort than \sys. Samoyed 
also provides fallbacks if an atomic function is too large, which \sys does not. 
Karma additionally considers asynchronous inputs, i.e., from interrupts, which \sys does not.

Capybara~\cite{capybara} is a hardware platform with a reconfigurable energy buffer, 
allowing for larger atomic regions to be executed when needed. HomeRun~\cite{homerun} also explores 
hardware support for atomicity in I/O events. 
Accumulative Display Updating~\cite{adp} 
explores relaxing atomicity constraints for long-running peripheral operations, such as updating displays, which 
does not meet the correctness definitions of \sys. 
Coati~\cite{coati}  and InK~\cite{ink} focus on event-driven execution and
are task-based. Tasks can be used with programmer effort to ensure freshness and consistency. 

\noindent\textbf{Correctness Reasoning}
Prior works~\cite{avis, formal-foundations, formal-periph} model the correctness of 
intermittent systems. Intermittent computing correctness is also similar to 
correctness of persistent 
memory~\cite{raad1, raad2, raad3, mp, mp2, persistent-linearizability} and to file system crash consistency
~\cite{fs-cc-models,yggdrasil,reachability,fscq}. Our notion 
of correctness follows most closely from~\cite{formal-foundations,fs-cc-models, reachability}, 
which define intermittent (or crashy) executions as correct if they are 
a refinement of some continuous (or non-crashy) execution. However, all these works define correctness 
in terms of memory consistency, and this continuous execution may pause arbitrarily. 
In this work, we show that these pauses introduce behaviour in the intermittent execution 
that is undesirable, even though memory is consistent. Our definitions of fresh and temporal consistency 
impose constraints on where these pauses are allowed.

\noindent\textbf{Transactions and Data-Centric Concurrency}
Atomic regions are similar to transactions~\cite{stm, tm}, though 
transactions use atomicity for concurrency, not 
timely processing of inputs. 
We draw the concept of consistent sets from the line of 
\emph{data-centric} concurrency control research
~\cite{data-centric-static, data-centric-ceze, data-coloring, data-centric-types,vaziri-atomic-set-serializability}. The 
data-centric approach is that programmers should indicate data that need 
to be synchronized, rather than onerously reasoning about operations and 
trying to place synchronization constructs accordingly. 
Data-coloring~\cite{data-coloring} is a programming model to automatically 
infer transaction placement for data consistency,  but it does so dynamically, requiring 
hardware support.
~\cite{data-centric-static,
data-centric-types} use types and static analysis to automatically infer 
synchronization operation placement, such as locks, that guarantees correctness for specified 
atomic sets, but the meaning of correctness is different. 
An atomic set is correct if it is serializeable; intermittent programs may experience timing violations even when memory safe.

\section{Conclusion and Future Work}
We present the properties of freshness and temporal consistency for
intermittent executions, linking the correct timing behaviour of an
intermittent execution to that of a continuous execution. Using these
definitions, we observe that atomicity can be used to provide correct
timing behaviour as well as memory consistency. To help enforce timing
constraints, we develop \sys, which is lightweight, and unlike prior
work does not require external hardware or complex reasoning about
real-time expiration or dataflow. \sys uses simple annotations
indicating which data should be fresh or temporally consistent to
infer atomic regions placement that automatically enforces correct
behaviour at runtime. The development of \sys additionally leads to several avenues 
of future work.

\paragraphb{Integration with Rust formalisms} \sys is the first intermittent 
computing toolchain to target Rust programs. Rust is an attractive target for intermittent 
computing systems as Rust programs 
are memory safe, reducing the likelihood that memory bugs will make a device inoperable after deployment 
to a remote environement. To prove that intermittence 
does not subvert the safety guarantees of Rust, however, future work should integrate intermittent computing 
semantics into existing Rust formalisms~\cite{rustbelt,rustbelt2}. 

\paragraphb{User Studies on Programmer Effort} We discuss the stategies and model the lines of code needed 
to use \sys, TICS, and Samoyed in Section~\ref{sec:prog-effort}. 
Truly comparing programmer effort and usability, however, needs to be done via user study.
\sys raises the usability questions of real-time versus implicit annotations, as well as annotations 
versus manually added regions. Carrying out a comprehensive user study on 
such features would benefit future system designers.

\paragraphb{Reasoning about Forward Progress} Along with memory consistency and 
timing-constraints, another key issue of correctness for intermittent 
computing systems is ensuring that programs can execute to completion.
Analyses that identify the minimum atomic regions 
necessary for correct execution, such as \sys's, can serve as a foundation for developing 
tools and formalisms to reason about forward progress and the inherent energy constraints of 
a program.

\section*{Acknowledgements}
We thank the anonymous reviewers for their feedback and Martin Rinard for 
shepherding this work. 
We also thank members of 
the Abstract Research Lab for their insightful comments on initial drafts.
This work was generously funded in part through National Science Foundation Award 2007998,
National Science Foundation CAREER Award 1751029, and the CMU CyLab Security \& Privacy Institute.

\bibliography{bib}
\clearpage
\appendix
\section{Syntax of the Modeling Language}
\label{app:syntax}
\[
\begin{array}{llcl}
\textit{values} & v & \bnfdef & n  \bnfalt \m{true}
\bnfalt \m{false} \bnfalt \&x \bnfalt \&a[i]
\\
\textit{expressions} & e & \bnfdef & x  \bnfalt v \bnfalt (a[i])  \bnfalt e_1 \bop e_2 \bnfalt \uop e
\\

\textit{instructions} & \instr & \bnfdef & 
\m{skip} \bnfalt x := e \bnfalt a[i] :=e \bnfalt *x := e 
\\
\textit{commands} & \cmd & \bnfdef &
\instr \bnfalt \ifthen{e}{\cmd_1}{\cmd_2} \bnfalt \cmd_1;\cmd_2\\
&& \bnfalt &  \elet{x}{e}{\cmd}  \\
&& \bnfalt & \elet{x}{f(v)}{\cmd} \bnfalt \elet{x}{\m{IN()}}{\cmd}\\
&& \bnfalt & \elet{\m{fresh}~x}{e}{\cmd}  \\
&&\bnfalt&  \m{let}~\m{consistent(n)}~x = e~\ft{in}~\cmd \\ 
&& \bnfalt&     \atomic{\aid, \omega}{\cmd} \\
\\
\textit{function decls} & \fdecls & \bnfdef& \cdot\bnfalt \fdecls, f(\ft{arg})= c; \m{ret}\;e
\end{array}
\]

\section{Taint-augmented  Semantics}
To formally define fresh and temporal consistency, we augment the
operational semantic rules with taint tracking of inputs. Note that
this is solely for formal definitions and proofs. We summarize the new
syntactic constructs below.
\[
\begin{array}{llcl}
\textit{Stack} & S & \bnfdef & \m{main} \bnfalt f \rhd \ell:\elet{x}{\hole}{\cmd}\rhd S
\\ &&\bnfalt & \hole;\cmd \rhd S
\\
\textit{input OPs} & \ins & \bnfdef & \vec{\tau}
\\
\textit{memory} & \nvmem^t & \bnfdef & \emptyset \bnfalt \nvmem^t, x\mapsto (v, \ins)
\\&&\bnfalt &\nvmem^t, a\mapsto [(v_1, \ins_1), \cdots, (v_n,\ins_n)]
\\
\textit{observations} & o & \bnfdef & \cdots \bnfalt \m{fresh}(f, \ell, \ins)
\bnfalt \m{cnst}(f, \ell, n, \ins)
\\&&\bnfalt& \m{use}(f, \ell, \tau)

\end{array}
\]
We write $S$ to denote execution contexts. Since the program always
starts from the $\m{main}$ function, the bottom of the stack is
$\m{main}$. $f \rhd \ell:\elet{x}{\hole}{\cmd}$ indicates that the
current function being executed is $f$, and once it returns, the
result will be bound to $x$ and the execution continues at
$\cmd$. $\hole;\cmd$ is the context for evaluating sequences.  Each
memory location now stores both the value and the input operations
that the value depends on. We write $\ins$ to denote the list of time
stamps where the input operations occur on the trace.  The execution
of an instruction could generate observations. Three observations
relate to the timeliness properties. $\m{fresh}(f, \ell, \ins)$ means
that in function $f$, line $\ell$, a freshness policy is declared and
the anti-dependent inputs are in $\ins$. $\m{cnst}(f, \ell, n, \ins)$
is the corresponding observation for a consistent policy. $\m{use}(f,
\ell, \tau)$ is generated when in function $f$ line $\ell$, the
variable declared to be fresh at time $\tau$ is used.

Figure~\ref{fig:taint-semantics} shows selected semantic rules.
\begin{figure}[htbp]
\noindent \framebox{$\tau_1, N^t_1, S_1, c_1 \SeqStepsto{O} \tau_2, N^t_2, S_2, c_2 $}

\begin{mathpar}
  \mprset{flushleft}
  \inferrule*[right=Call]{    \ft{getTnt}(N^t, v) = \ins
    \\ \ft{getUse}(N^t, v) = o
  }{
    \tau, N^t, S, \ell: \elet{x}{f(v)}{\cmd}
    \\ \SeqStepsto{o}
    \tau+1,[\ft{arg}\mapsto (v,\ins)]\rhd N^t,
    \\\\\qquad f\rhd\ell: \elet{x}{\hole}{\cmd}\rhd S, \fdecls(f)
  }
  \and
  \inferrule*[right=Ret]{
    \ft{getTnt}(N^t, v) = \ins
    \\ \ft{getUse}(N^t, v) = o
  }{
     \tau, [\ft{arg}\mapsto \_]\rhd N^t,
     f\rhd\ell: \elet{x}{\hole}{\cmd}\rhd S, \m{ret}\, v
     \\ \SeqStepsto{o}
     \tau+1, [x \mapsto (v,\ins)]\rhd N^t, S, \cmd;\m{drop}
   }
   \and
   \inferrule*[right=Fresh]{
     \m{eval}(N^t, e) = (v, \ins)
     \\\ft{top}(S)= f
     \\\ft{getUse}(N^t, e) = O
  }{
    \tau, N^t, S, \ell: \elet{\m{fresh}~x}{e}{\cmd}
    \\ \SeqStepsto{O,\m{fresh}(\m{top}(S), \ell, \ins)}
    \tau+1,[x^{f, \ell, \tau} \mapsto (v,\ins)]\rhd N^t,S, \cmd;\m{drop}
   }
   \and
   \inferrule*[right=Consistent]{
     \m{eval}(N^t, e) = (v, \ins)
     \\\ft{top}(S)= f
     \\\ft{getUse}(N^t, e) = O
  }{
    \tau, N^t, S, \ell: \elet{\m{consistent(n)}~x}{e}{\cmd}
    \\ \SeqStepsto{O,\m{cnst}(\m{top}(S), \ell, n, \ins)}
    \tau+1,[x (v,\ins)]\rhd N^t,S, \cmd;\m{drop}
  }
  \and
  \inferrule*[right=In]{    
  }{
    \tau, N^t, S, \ell: \elet{x}{\m{IN}()}{\cmd}
    \\ \SeqStepsto{}
    \tau+1,[x \mapsto (\m{in}(\tau),\tau)]\rhd N^t,S, \cmd;\m{drop}
  }
  \and
\end{mathpar}
\caption{Augmented semantic rules}
\label{fig:taint-semantics}
\end{figure}

\section{Timeliness definitions}
\label{app:correctness-defs}

\begin{defn}[Freshness]
  We say that a program $p$ satisfies a freshness constraint declared in function $f$ at location $\ell$ if for all traces
  $T=  \tau, \emptyset, \m{main}, \cmd \MSeqStepsto{O}$,
  any segment $T'$ of $T$ s.t. $T'$
  \begin{itemize}
  \item includes an observation  $\m{fresh}(f, \ell, \ins)$ at $\tau$,
  \item begins with the earliest time stamp in $\ins$,
  \item includes all the observations $\m{use}(f, \ell, \tau)$
  \item ends with the last  $\m{use}(f, \ell, \tau)$ observation
  \end{itemize}
  it is the case that $T'$ is nested within a $\m{begin\_atom}(\aid)$
  and an end $\m{end\_atom}(\aid)$, with no other entering/exiting atomic region operations. 
\end{defn}

\begin{defn}[Consistency]
  We say that a program $p$ satisfies consistency constraints with ID(n) declared in function $f$ if for all traces
  $T=  \tau, \emptyset, \m{main}, \cmd \MSeqStepsto{O}$, exists $T'$ s.t. $T'$
  \begin{itemize}
  \item includes a call to function $f$ and return from $f$
  \item includes consistency observations:
    \\$\m{cnst}(f, \ell_1, n, \ins_1)$...$\m{cnst}(f, \ell_k, n, \ins_k)$
    \end{itemize}
  any segment $T''$ of $T$ s.t. $T''$
  \begin{itemize}
  \item begins with the earliest time stamp in $\bigcup_1^n\ins_i$
  \item ends with the last  time stemp in $\bigcup_1^n\ins_i$
  \end{itemize}
  it is the case that $T'$ is nested within a $\m{begin\_atom}(\aid)$
  and an end $\m{end\_atom}(\aid)$, with no other entering/exiting atomic region operations. 
\end{defn}

\section{Atomic region checking}
\label{app:atomic-checking}
Key rules for atomic region checking are shown in
Figure~\ref{fig:atimic-check}.  Note that when a function is called,
the function body is checked, and thus these set of rules traverse all
the execution paths. Since we don't have recursive functions, the
traversal is guaranteed to terminate.

\begin{figure}[tbhp]
\begin{mathpar}
  \inferrule*[right=Instr-N]{
    \m{lookup}(\pdecls,\pmap,f; \prov,\ell{:}\instr) = \m{none}
  }{   \fdecls;\pdecls,\pmap; f; \prov;  {\pol}s; \aid \Vdash \ell{:}\instr : {\pol}s
  }
  \and
  \inferrule*[right=Instr-S]{
    \m{lookup}(\pdecls,\pmap,f; \prov,\ell{:}\instr) = \aid
    \\  {\pol}s' = \m{remove}({\pol}s, f; \prov,\ell{:}\instr)
  }{   \fdecls;\pdecls,\pmap; f; \prov;  {\pol}s; \aid \Vdash \ell{:}\instr : {\pol}s'
  }
  \and
  \inferrule*[right=Call-N]{
  \m{lookup}(\pdecls,\pmap,f; \prov,\ell{:}\instr) = \m{none}
    \\ \fdecls;\pdecls,\pmap; g; (f, \ell)::\prov;  {\pol}s; \aid \Vdash \fdecls(g): {\pol}s'
    \\  \fdecls;\pdecls,\pmap; f; \prov;  {\pol}s'; \aid \Vdash \cmd: {\pol}s''
  }{    \fdecls;\pdecls,\pmap; f; \prov;  {\pol}s; \aid \Vdash
     \ell:\elet{x}{g(v)}{\cmd}: {\pol}s''
  }
  \and
  \inferrule*[right=Atomic]{
    \fdecls;\pdecls,\pmap; f;  \prov; \m{lookup}(\pdecls,\pmap,\aid); \aid  \Vdash \cmd : {\pol}s
    \\\m{isEmpty}({\pol}s)
}{   \fdecls;\pdecls,\pmap; f; \prov; \emptyset; \cdot \Vdash \atomic{\aid, \omega}{\cmd} : \m{ok}
  }
\end{mathpar}
\caption{Atomic region checking rules}
\label{fig:atimic-check}
\end{figure}

We prove the following lemma showing that the static checking
over-approximates the real trace (due to execution taking different
branches).
\begin{lem}[Symbolic check matches trace]~\label{lem:atomic-check}
  ~\\
  \begin{enumerate}
  \item If $\fdecls;\pdecls,\pmap; f; \prov; \emptyset; \cdot \Vdash
    \cmd: \m{ok}$, and $\m{callStack}(S)= f::\prov$ and $ T = (\tau,
    \nvmem^t, S, \cmd)\MSeqStepsto{O}(\tau_1, \nvmem^t_1, S_1,
    \m{skip})$, then for all policies in $\pdecls$ are satisfied on $T$.

    \item If $\fdecls;\pdecls,\pmap; f; \prov; \pol{s}; \aid \Vdash
      \cmd: \pol{s}'$, and $\m{callStack}(S)= f::\prov$ and $ T =
      (\tau, \nvmem^t, S, \cmd)\MSeqStepsto{O}(\tau_1, \nvmem^t_1,
      S_1, \m{skip})$, then all actions in $\pdecls$, only actions in
      $\pol{s}'\backslash\pol{s}$ are performed on $T$. Actions in
      $\pol{s}'\backslash\pol{s}$ not performed on $T$ cannot be
      reached on $T$ (need to explore a different path).
  \end{enumerate}
\end{lem}
\begin{proof}(Sketch)
  By induction over the derivation. 
  \end{proof}

\section{Checking summary and policy declaration}
\label{sec:time-checking}

Figure~\ref{fig:tmap-checking} summarizes taint summary and policy checking rules. 

\begin{figure*}[t!]
\begin{mathpar}\mprset{flushleft}
  \inferrule*[right=Input]{
     \fdecls;\pdecls,\fsums;c; f; M; I\cup(x \hookleftarrow ((f, \ell),\m{local}(\ell)) 
     \Vdash \cmd: M'; I'
 }{ \fdecls;\pdecls,\fsums;c; f; M; I
    \Vdash  \ell:\elet{x}{\m{IN}()}{\cmd}: M' \backslash x; I' \backslash x
  }
  \and
  \inferrule*[right=Let]{
    \ft{checkUse}(\pdecls, e)
    \\
    \ft{inputs} = I(e) \\
    \fdecls;\pdecls,\fsums; c; f; M; I\cup(x \hookleftarrow \ft{inputs}) 
    \Vdash \cmd: M'; I'
  }{ \fdecls;\pdecls,\fsums; c; f; M; I
    \Vdash  \ell:\elet{x}{e}{\cmd}: M' \backslash x; I' \backslash x
  }
  \and
  \inferrule*[right=Call-nr]{
    \ft{checkUse}(\pdecls, v)\\
    v~\mbox{is not a reference}\\
    \ft{inputs} = I(v) \\
    \fsums(g) = s\\
    \ft{inputs} \subseteq s(\m{call}, f, \ell, \m{arg})\\
    \ft{outputs} = s(\m{local}, \m{ret}) \cup s(\m{call}, f, \ell, \m{ret})\\
    \ft{outputs'} = \ft{outputs}[\ft{fromTp}\mapsto \m{retBy}(g, \ell)] \\
    \fdecls;\pdecls,\fsums; c; f; M; I\cup(x \hookleftarrow \ft{outputs'})
    \Vdash \cmd: M'; I'
  }{ \fdecls;\pdecls,\fsums; c; f; M; I
    \Vdash  \ell:\elet{x}{g(v)}{\cmd}: M' \backslash x; I' \backslash x
  }
  \and
  \inferrule*[right=Call-r]{
    \ft{checkUse}(\pdecls, v)\\
    \ft{inputs} = I(y) \\
    \fsums(g) = s\\
    \ft{inputs} \subseteq s(\m{call}, f, \ell, \m{arg})\\
    \ft{outputs} = s(\m{local}, \m{ret}) 
    \cup s(\m{call}, f, \ell, \m{ret})\\
    \ft{outputs'} = \ft{outputs}[\ft{fromTp}\mapsto \m{retBy}(g, \ell)) \\
    \ft{pbr} =s(\m{local}, \m{\&arg})
    \cup s(\m{call}, f, \ell, \m{\&arg})\\
       \ft{pbr'} = \ft{pbr}[\ft{fromTp}\mapsto \m{pbr}(f, \ell)] \\
    \fdecls;\pdecls,\fsums; c; f; M;
    (I\cup(x \hookleftarrow \ft{outputs'}))[y \hookleftarrow \ft{pbr'}]
    \Vdash \cmd: M'; I'
  }{ \fdecls;\pdecls,\fsums;c; f; M; I
    \Vdash  \ell:\elet{x}{g(\&y)}{\cmd}: M' \backslash x; I' \backslash x
  }
  \and
  \inferrule*[right=Ret]{
     \ft{checkUse}(\pdecls, e)\\
    \ft{inputs} = I(e) \\
    \fsums(f) = s\\
    \ft{inputs} \subseteq s(\m{call}, c, \m{ret})\cup  s(\m{local}, \m{ret})\\
  }{ \fdecls;\pdecls,\fsums; c; f; M; I
    \Vdash  \ell:\m{ret}\; e : \m{done}
  }
  \and
  \inferrule*[right=Assign]{
    \ft{checkUse}(\pdecls, e)\\
    \ft{inputs} = I(e) \\
    I'= I[x \hookleftarrow \ft{inputs}]
  }{ \fdecls;\pdecls,\fsums; c; f; M; I \Vdash  \ell: x := e : M'; I'
      }
  \and
  \inferrule*[right=Assign-Ref]{
    M' = M[x \mapsto M(e)]\\
    \ft{checkUse}(\pdecls, e)\\
    \\\\
    \ft{inputs} = I(e) \\
    I'= I[x \hookleftarrow \ft{inputs}]
    \\\\        
    \mbox{if}~ M(x) = \ft{arg}: ~\fsums(f) = s\\
    \quad \ft{inputs} \subseteq s(\m{local}, \&\m{arg})
    \cup s(\m{call}, c, \m{\&arg} )
  }{ \fdecls;\pdecls,\fsums; c; f; M; I \Vdash  \ell: *x := e : M'; I'
  }
  \and
  \inferrule*[right=Let-fresh]{
    \ft{inputs} = I(e) \\
    \ft{callChain}(\fsums,\ft{inputs})\subseteq \pdecls(\m{fresh}, f, \ell).\m{inputs}
    \\\\
    \fdecls;\pdecls,\fsums; c; f; M\cup(x\mapsto M(e)); I\cup(x \hookleftarrow \ft{inputs}) 
    \Vdash \cmd: M'; I'
  }{ \fdecls;\pdecls,\fsums; c; f; M; I
    \Vdash  \ell:\elet{\m{fresh}\;x}{e}{\cmd}: M' \backslash x; I' \backslash x
  }

  \and
  \inferrule*[right=Let-consistent]{
    \ft{inputs} = I(e) \\
    \ft{callChain}(\fsums,\ft{inputs})\subseteq \pdecls(\m{consistent}, f, \ell).\m{inputs}
    \\\\
    \fdecls;\pdecls,\fsums; c; f; M\cup(x\mapsto M(e)); I\cup(x \hookleftarrow \ft{inputs}) 
    \Vdash \cmd: M'; I'
  }{ \fdecls;\pdecls,\fsums; c; f; M; I
    \Vdash  \ell:\elet{\m{consistent(n)}\;x}{e}{\cmd}: M' \backslash x; I' \backslash x
  }
  \and
  \inferrule*[right=Atomic]{
    \fdecls;\pdecls,\fsums; c; f; M; I \Vdash \cmd : M'; I'
  }{    \fdecls;\pdecls,\fsums; c; f; M; I  \Vdash \atomic{\aid, \omega}{\cmd} : M'; I'
  }
\end{mathpar}

\begin{mathpar}
  \inferrule*{
    \pdecls(\m{fresh},f) = F  \\
    \forall (f, \ell_i) \in \m{dom}(F), s.t. \fdecls(f, \ell_i).\m{var} \in \m{fv}(e),
    (f, \ell) \in F(f, \ell_i).\m{uses}
  }{
    \ft{checkUse}(\pdecls, e)
    }
  \end{mathpar}
\caption{Checking taint and use policies}
\label{fig:tmap-checking}
\end{figure*}

We define well-formedness of a configuration w.r.t. a trace $T$ as follows.
$\fdecls;\pdecls,\fsums; T \Vdash (N^t, S, \cmd) :\m{ok}$
iff exists $M$, $\cc$; $f$, $I$ s.t.
$\fdecls;\pdecls,\fsums; \cc; f; M; I \Vdash \cmd : M'; I'$, 
$\m{isTop}(S, f, \cc)$, $\m{overapprox}(M, N^t|\m{fv}(\cmd))$,
$\m{overapprox}(\fsums, I, N^t|\m{fv}(\cmd), T)$,
and $\fdecls;\pdecls,\fsums; M'; I' \Vdash S : \m{ok}$,

$\m{overapprox}(M, N^t)$ is true if the may alias in $M$ over
approximate the may alias in the memory.
$\m{overapprox}(\fsums, I,N^t, T)$
is true if the tainting from inputs in $I$ over approximates
the taint information in the memory. Here, $I$ includes only segments
of the taint provenance, on the other hand, the taint information
stored in $N^t$ includes timestamp of the taint operation. We further
define a function to recover the call trace from $I$ and $\fsums$. It
essentially traverses the $\ft{fromTp}$ and stops at a local input
operation. From the memory, we can extract the call chain from $S$ in
the configuration at that time stamp on $T$.

\begin{lem}[One step preservation]\label{lem:one-step}
  Given a
  trace $T$ s.t. $T\SeqStepsto{}(\tau, N^t, S, \cmd)$, and
  $\fdecls;\pdecls,\fsums; T \Vdash (N^t, S, \cmd) :\m{ok}$ and
  $(\tau, N^t, S, \cmd)\SeqStepsto{} (\tau_1, N^t_1, S_1, \cmd_1)$
  then
  $\fdecls;\pdecls,\fsums; T\SeqStepsto{}(\tau, N^t, S, \cmd) \Vdash N_1^t, S_1, \cmd_1: \m{ok}$. 
\end{lem}
\begin{proof}(sketch)
    By examining all the semantic rules. 
\end{proof}

\begin{lem}[Trace preservation]
  \label{lem:preservation}
  If a program is checked to be $\m{ok}$, a trace starting from
  initial state and function $\m{main}$ all intermediate
  configurations are also $\m{ok}$.
\end{lem}
\begin{proof}(sketch)
    By induction over the length of $T$ and use Lemma~\ref{lem:one-step}.
\end{proof}

\section{Correctness Theorem}

\begin{lem}[Policy conservative]\label{lem:policy-cons}
If $\fdecls;\pdecls,\fsums \Vdash \fsums :\m{ok}$, and $T$ is trace
starting from initial state and function $\m{main}$ then any action
that is related to any policy $\pid$ observed on a trace $T$, is
(after replacing time stamps with call chain for inputs) included in
$\pdecls(\pid)$.
\end{lem}
\begin{proof}(sketch)
  This follows from Lemma~\ref{lem:preservation}. For freshness
  policies, if a trace include a declared fresh variable $x$, then the
  trace must include a configuration that has as its command,
  $\elet{\m{fresh}\, x}{e}{\cmd}$. Given the configuration is checked
  $\m{ok}$, by Lemma~\ref{lem:preservation} we know that all the
  inputs that $e$ depends on is a super set as on the trace. Further,
  the checking rule checks those inputs are declared in $\pdecls$.
  Therefore, all the relevant actions as observed on the trace is
  included in the policy spec. The uses of $x$ will be included in the
  policy as the checking rules ensure that all uses of $x$ are
  included in the policy spec.
  
  Similar arguments can be made for  and consistency.
\end{proof}

\begin{thm}
  Given a program $p$ consisting of functions in $\fdecls$,
  $\fdecls;\pdecls,\fsums \Vdash \fsums :\m{ok}$ and
  $\fdecls;\pdecls,\pmap; \m{main}; \emptyset;\cdot \Vdash
  \fdecls(\m{main}) : \m{ok}$, then $p$ satisfies all the specified
  polices.
\end{thm}
\begin{proof}(sketch)
  This follows from Lemma~\ref{lem:policy-cons} and
  Lemma~\ref{lem:atomic-check}. First all operations are covered by
  the policy, then all atomic regions are shown to wrap actions within
  one policy completely.  That is if a program $p'$ pass the check,
  then all input operations that a fresh annotated variable depends
  on, as well as any uses of the variable, will be in the same atomic
  region. Any input operations that any item in a consistent set
  depends on will also be in the same region. As the committed
  execution of a region never experiences a power-failure, the
  committed execution always has the same timing-behaviour as a
  continuous execution.
\end{proof}

\section{Correctness of inference algorithm}

\paragraphb{Input dependence map}
Building the input dependence map largely works as described in the checking rules. 
This area of the analysis is the primary beneficiary of Rust being the target language.
The analysis can assume no mutable globals, and that references are owned --- 
there is only one copy of a mutable reference. These assumptions simplify taint analysis.
Static reasoning about tainted globals and pointers otherwise needs to be strongly conservative. 
In particular, Rust ownership means that if taint is stored into a 
pass-by-reference parameter, that reference is written to, so it won't alias with any other references passed in. 

\paragraphb{Selection of the Goal Function}
\milijana{If we think it's a better way to present the algo, I could write inference rules for this pretty quick}
 the above rules. Basically, the algorithm is recording the nest of functions 
 an operation in the region resides in. 
Selecting the deepest function that still includes everything will pass the check, but with a 
smaller region size than a shallower function. Consider the 
checking rule for functions, rule \rulename{Func}. If there is a command $x := g()$, where $g() = y := f()$, and $\cmd_f$ 
contains the input operations, then 
if $ y := f()$ checks $\m{ok}$, so will $x := g()$.

\paragraphb{Inserting the region}
The start and end of an atomic region are chosen to be points that dominate and post-dominate, respectively, 
all operations that need to be in the same region to pass the check. \milijana{Say something about panicking?}

\section{Formal Semantics of the JIT + Atomics Execution Model}
\label{app:semantics}
A state is of the form $(\timestamp, \context, S, \vmem, \cmd)$.
$\timestamp$ denotes the logical time of the state, $\context$ is the 
saved execution context, $\nvmem$ is the non-volatile memory of the system, $S$ is the stack, and $\cmd$ is the 
command to be executed.
State transitions are of the form $(\timestamp, \context, \nvmem, S, \cmd) 
\stepsto (\timestamp', \context', \nvmem', S', \cmd')$

A context can either be a JIT context $\jctx$ or an atomic context $\actx$. 
Whichever type of context is currently set governs the behaviour on checkpoints, low power 
triggers, and reboots. The JIT context is of the form $\m{jit}(S,\cmd)$,
where $S$ and $\cmd$ are the execution context (stack and command) saved at a checkpoint, and 
the atomic context is of the form $\m{atom}(\ulog, S, \cmd, \depth)$, where 
$S$ and $\cmd$ have the same meaning as for a JIT checkpoint, $\ulog$ is 
the non-volatile data that must be saved, and $\depth$ is a counter for 
the nesting of atomic regions. 

Note that an execution will always begin with $\jctx$, as the context 
is a piece of nonvolatile memory that is statically initialized to refer 
to the beginning of the program $(\cmd_0, S_0)$.

\begin{mathpar}
\inferrule[JIT-LowPower]{ \mt{PowerLow} \\ \m{pick}(n)}{
  (\timestamp,\jctx, \nvmem, S, \cmd) 
 \Stepsto{} (\timestamp+1, \m{jit}(S, \cmd), \nvmem, S, \m{reboot}(n)) 
}
\and
\inferrule[Atom-LowPower]{\mt{PowerLow} \\ \m{pick}(n)}{
  (\timestamp,\actx, \nvmem, S, \cmd) 
 \Stepsto{}  (\timestamp+1, \actx, \nvmem, S, \m{reboot}(n)) 
}
\and 
\inferrule[JIT-Reboot]{ \jctx = \m{jit}(S, \cmd)}{
(\timestamp, \jctx, \nvmem, S', \m{reboot}(n)) 
  \Stepsto{}  ( \timestamp+n, \jctx, \nvmem, S, \cmd) 
}

\and 
\inferrule[Atom-Reboot]{ }{
(\timestamp, \m{atom}(\ulog, S, \cmd, \depth), \nvmem, S', \m{reboot}(n)) 
\Stepsto{} \\(\timestamp+n,  \m{atom}(\ulog, S, \cmd, 0), 
\nvmem \lhd \ulog, S, \cmd) 
}
\end{mathpar} 
\paragraphb{Power Failures and Reboots} 
On receiving a low power signal, a system in JIT mode 
saves the current volatile memory and command to the context and transitions 
to the reboot command (rule \rulename{JIT-LowPower}). If the system was in Atomic
 mode, however, it immediately 
transitions to reboot (rule \rulename{Atom-LowPower}). In both cases, the reboot command is parameterized with 
an $n$ picked at random.  

If a reboot command executes in JIT mode, then the system updates the command and stack with 
the execution state stored in the context and continues executing the command. If the system 
reboots when in atomic mode, it applies the undo log to nv memory, updates the 
command and stack with the values from the 
context, and sets the 
atomic depth counter to zero. This update ensures that the atomic depth counter 
will remain consistent upon re-execution of any nested atomic starts and ends. 
In both cases, the timestamp $\timestamp$ is updated with the picked n. 
This randomly large update to $n$ captures how power can be off 
for an arbitrary period of time. Continuing a sequence of input operations on 
the intermittent system with the new timestamp may not match the desirable 
behaviour on the continuously powered system.

\begin{mathpar}
\inferrule[Atom-Start-Outer]{\actx = \m{atom}(\nvmem |_{\omega},S, \cmd, 0)  }{
  (\timestamp, \jctx, \nvmem, S, \astart{\omega};\cmd) 
  \Stepsto{}  (\timestamp+1, \actx, \nvmem, S, \cmd) 
}
\and

\inferrule[Atom-Start-Inner]{ }{
  (\timestamp, \actx, \nvmem, S, \astart{\omega};\cmd) 
  \Stepsto{} \\ (\timestamp+1, \actx[\depth \leftarrow \depth+1], \nvmem, S, \cmd) 
}
\and
\inferrule[Atom-End-Outer]{ \actx = \m{atom}(\nvmem_a,S_a, \cmd_a, 0)
\\ \jctx = \m{jit}(\emptyset, \emptyset)}{
  (\timestamp, \actx, \nvmem, S, \aend;\cmd) 
 \Stepsto{} (\timestamp+1, \jctx, \nvmem, S, \cmd) 
}
\and 
\inferrule[Atom-End-Inner]{\actx = \m{atom}(\nvmem_a,S_a, \cmd_a, \depth)\\ \depth > 0 }{
  (\timestamp, \actx, \nvmem, S, \cmd) 
 \Stepsto{}  (\timestamp+1, \actx[\depth \leftarrow \depth -1], \nvmem, S, \cmd) 
}
\end{mathpar}
\paragraphb{Atomic region transitions}

Rule \rulename{Atomic-Start-Outer} describes the behaviour when 
the system transitions into Atomic mode from JIT mode. 
The context is
  switched to an undo log checkpoint containing the command to be executed $\cmd$, 
  the (often volatile) stack memory $S$, the nonvolatile data 
  that must be saved based on the potentially inconsistent variables of the 
  region ($\nvmem|_{\omega}$), and an atomic depth counter $\depth$, set to 0. 
  If the system encounters another atomic region start while already in atomic 
  mode (rule \rulename{Atomic-Start-Inner}), then that counter is incremented, but 
  otherwise the command is a no-op. If the counter is greater than zero when 
  encountering an $\aend$ command, then the system similarly decrements the depth counter 
  but otherwise does nothing (\rulename{Atomic-End-Inner}). If the counter 
  is zero, and the system is ending the outermost atomic region (\rulename{Atomic-End-Outer}), 
  then the system switches the context to an empty JIT context. This behaviour 
  is safe as if the system experiences a low power signal, then the system 
  will populate the JIT checkpoint.


\section{Building an Input Map}
\label{app:taint-tracking}
We show \sys's algorithm for building an input map.
It starts by computing a static taint analysis of any input operations, 
building a map of a variable definition to the call chain of any input operation 
on which it depends. We show 
the pseudo-code in Algorithm~\ref{alg:buildmap}. The top-level algorithm in inter-procedural, using 
the function \rulename{Track} to compute taint propagation within a single function. 
\rulename{Track} takes as input four parameters: \alg{currInst}, which is the instruction from 
which to start propagating taint, \alg{iOp}, which is the source of taint into the current 
function, \alg{tMap}, the taint Map being built, and \alg{caller}, which has a value only 
if taint was passed in from a calling function. The algorithm for local taint propagation is 
standard and is omitted for space. The key features are a) it inserts any definitions that 
are data or control dependent on \alg{iOp} into the taint map, b) it is context-sensitive, 
propagating taint to all callers if taint was generated within the local function, but only 
to \alg{caller} if taint was passed in, and c) it takes advantage of Rust's mutability to know 
that a written-to pointer cannot alias with any other pointers.
\rulename{Track} returns a summary of taint 
propagation within the local function, namely how taint was passed \emph{in} to the local function, 
how it propagated \emph{out}, and what \emph{type} of propagation it was. 
 We explain 
the possible types as we step through the top-level function \rulename{buildInputDeps.}

\begin{algorithmic}[1]
  \footnotesize
  \Function{buildInputDeps}{\alg{ Cmd}}
  
  \State \alg{inputs \gets Cmd.findInputInsts()}

  \ForAll{\alg{iOp \in inputs}}
  
  \State \alg{first \gets Summary\{in:null, out:iOp, type:input\}}
  \State \alg{toExplore.append(first)}
  
  \While{\alg{s \gets toExplore.next()}}
  
  \If{\alg{s.type == input}}
  
  \State \alg{ipFlow \gets track(s.out, s.out, tntMap, null)} \Comment{start from the input}
  \ElsIf{\alg{s.type = return}}
  \State \alg{ipFlow \gets track(s.out, s.out, tntMap, null)}\Comment{ret is now taint src}
  \ElsIf{\alg{s.type = passbyref}}
  \State \alg{tSrc \gets s.out.call} \Comment{callinst is now taint src}
  \State \alg{start \gets s.out.next}\Comment{start from next use of ref}
  \State \alg{ipFlow \gets track(start, tSrc, tntMap, null)}
  
    \ElsIf{\alg{s.type = argument}}

    \State \alg{cllr \gets s.out.call}\Comment{only this context is tainted}
    
    \State \alg{ipFlow \gets track(s.out, s.in, tntMap, cllr)}\Comment{use cllrs taint source}
    \EndIf
    \State \alg{\forall nxt \in ipFlow, tntMap[nxt.out] \gets nxt.in}
  \State \alg{toExplore.add(ipFlow)}
  \EndWhile
  \EndFor
  \State\Return{\alg{tntMap}}
  \EndFunction

  \end{algorithmic}
  \label{alg:buildmap}
  \captionof{algorithm}{Build an input-dependence map}

\rulename{BuildInputDeps} starts by finding any calls to input functions, which were indicated by 
the programmer. It then computes taint propagation for each call. 
The inter-procedural taint propagation summaries are 
stored into the list \alg{ipFlow}. Taint can propagate between functions if it 
is stored into a \alg{return} instruction, stored into a \alg{pass-by-reference} parameter, 
or if it is used as an \alg{argument}. Lines 9-10 describe the return case.
If taint is returned into a function $F$, then the algorithm treats the callsite 
as both the next source of taint as well as the starting point for calling \rulename{Track} on $F$. 
If taint is propagated to $F$ through pass-by-ref (lines 11-14), the callsite 
is similarly the source of taint in $F$, but taint propagation starts from 
the next use of that reference. If taint is propagated into a called function $F1$ (lines 15-17), 
then the tainted source is the same source as that of the caller $F$, and taint propagation 
starts from uses of the argument in $F1$. The \alg{caller} parameter is set 
to $F$, so that \rulename{Track} propagates taint context-sensitively.
At line 19, the algorithm
adds an entry into \alg{tntMap} for each flow, mapping the out edge to the tainted source. This 
map structure allows call chains of a tainted flow to be retrieved. 
For example, if a function $F_{in}$ calls an 
input operation $x := \mt{IN()}$ and returns it to $F2$, the entry 
$[ret := call F_{in}] \mapsto x := IN()$ will 
be added to the map. If F2 stores the returned value into a variable $y$, the entry 
$[y := ret] \mapsto ret := call F_{in}$ will be added to the map. Chaining map lookups 
will retrieve the entire sequence from the definition to the original input operation.

\end{document}